\documentclass[11pt]{article}
\usepackage{pgfplots}
\usepackage{fullpage}
\pgfplotsset{compat=1.7}
\usepackage{subcaption}
\usepackage[T1]{fontenc}
\usepackage{lmodern}
\usepackage{amsmath}
\usepackage{amssymb}
\usepackage{amsthm}
\usepackage{color,soul}
\usepackage{xcolor}
\usepackage{graphicx}
\usepackage{tcolorbox}
\usepackage[normalem]{ulem}
\usepackage{float}
\usepackage{thmtools}
\usepackage{hyperref}
\usepackage{cleveref}
\usepackage{microtype}
\usepackage{caption}
\usepackage{bbm}
\usepackage{hyperref, color}
\hypersetup{colorlinks=true,citecolor=blue, linkcolor=blue, urlcolor=blue}
\usepackage[linesnumbered,boxed,ruled,vlined]{algorithm2e}
\usepackage{bm}
\usepackage{bbm}
\usepackage[numbers]{natbib}
\usepackage{xcolor}
\usepackage{enumerate} 
\usepackage{enumitem}
\usepackage{tabularx}
\usepackage{array}
\usepackage{cleveref}
\usepackage{tcolorbox}
\usepackage[thinlines]{easytable}
\usepackage{tikz}

\usetikzlibrary{decorations.pathmorphing, arrows.meta}
\newcolumntype{L}[1]{>{\raggedright\arraybackslash}p{#1}}
\newcolumntype{C}[1]{>{\centering\arraybackslash}m{#1}}
\newcolumntype{R}[1]{>{\raggedleft\arraybackslash}p{#1}}

\usepackage{makecell}
\usepackage{footnote}
\usepackage[ruled, vlined]{algorithm2e}
\makesavenoteenv{tabular}

\newcommand{\DTV}[2]{d_{\mathrm{TV}}\left({#1},{#2}\right)}

\newcommand{\e}{\mathrm{e}}

\renewcommand{\epsilon}{\varepsilon}

\newcommand{\oracle}{T}

\newtheorem{theorem}{Theorem}[section]

\newtheorem*{claim*}{Claim}
\newtheorem{condition}[theorem]{Condition}

\newtheorem{lemma}[theorem]{Lemma}

\newtheorem{corollary}[theorem]{Corollary}
\theoremstyle{definition}

\newtheorem{definition}[theorem]{Definition}
\newtheorem{problem}[theorem]{Problem}
\newtheorem{remark}[theorem]{Remark}
\newtheorem*{remark*}{Remark}
\newtheorem{assumption}{Assumption}

\def\Pr{\mathop{\mathbf{Pr}}\nolimits}

 \newcommand{\tuple}[1]{\left(#1\right)} 
 
 \newcommand{\tp}{\tuple}

\newcommand{\dvg}[3]{D_{#3}\left({{#1}\,\|\,{#2}}\right)}
\newcommand{\hdvg}[3]{\hat{D}_{#3}\left({{#1}\,\|\,{#2}}\right)}
\newcommand{\tdvg}[3]{\tilde{D}_{#3}\left({{#1}\,\|\,{#2}}\right)}

\def\*#1{\mathbf{#1}} 
\def\+#1{\mathcal{#1}} 
\def\-#1{\mathrm{#1}} 

\usepackage{todonotes}
\usepackage{xifthen}

\renewcommand{\Pr}[2][]{ \ifthenelse{\isempty{#1}}
  {\mathbf{Pr}\left[#2\right]} {\mathbf{Pr}_{#1}\left[#2\right]} } 
\newcommand{\E}[2][]{ \ifthenelse{\isempty{#1}}
  {\mathbf{\mathbf{E}}\left[#2\right]}
  {\mathbf{\mathbf{E}}_{#1}\left[#2\right]} }
  \newcommand{\Var}[2][]{ \ifthenelse{\isempty{#1}}
  {\mathbf{\mathbf{Var}}\left[#2\right]}
  {\mathbf{\mathbf{Var}}_{#1}\left[#2\right]} }

\crefname{theorem}{Theorem}{Theorems}
\crefname{observation}{Observation}{Observations}
\crefname{claim}{Claim}{Claims}
\crefname{condition}{Condition}{Conditions}
\crefname{algorithm}{Algorithm}{Algorithms}
\crefname{property}{Property}{Properties}
\crefname{example}{Example}{Examples}
\crefname{fact}{Fact}{Facts}
\crefname{lemma}{Lemma}{Lemmas}
\crefname{corollary}{Corollary}{Corollaries}
\crefname{definition}{Definition}{Definitions}
\crefname{remark}{Remark}{Remarks}
\crefname{proposition}{Proposition}{Propositions}
\crefname{equation}{equation}{equations}
\crefname{enumi}{}{}
\crefname{enumii}{}{}
\crefname{enumiii}{}{}
\crefname{enumiv}{}{}
\creflabelformat{enumi}{#2#1#3}
\creflabelformat{enumii}{#2#1#3}
\creflabelformat{enumiii}{#2#1#3}
\creflabelformat{enumiv}{#2#1#3}

\newboolean{DoubleBlind}
\setboolean{DoubleBlind}{false}

\title{On approximating the $f$-divergence between two Ising models}

\author{\ifthenelse{\boolean{DoubleBlind}}{Author(s)}
{Weiming Feng \footnote{School of Computing and Data Science, The University of Hong Kong.\\ \makebox[1.78em][l]{}Emails: \texttt{wfeng@hku.hk} and \texttt{fyc0130@connect.hku.hk}.
} \and Yucheng Fu \footnotemark[1]
}}

\date{}

\begin{document}

\maketitle

\begin{abstract}
    The $f$-divergence is a fundamental notion that measures the difference between two distributions. In this paper, we study the problem of approximating the $f$-divergence between two Ising models, which is a generalization of recent work on approximating the TV-distance. Given two Ising models $\nu$ and $\mu$, which are specified by their interaction matrices and external fields, the problem is to approximate the $f$-divergence $\dvg{\nu}{\mu}{f}$ within an arbitrary relative error $\e^{\pm \varepsilon}$. For $\chi^\alpha$-divergence with a constant integer $\alpha$, we establish both algorithmic and hardness results. The algorithm works in a parameter regime that matches the hardness result. Our algorithm can be extended to other $f$-divergences such as $\alpha$-divergence, Kullback-Leibler divergence, R\'enyi divergence, Jensen-Shannon divergence, and squared Hellinger distance.
\end{abstract}

\section{Introduction}
Let $\nu$ and $\mu$ be two distributions with the same support $\Omega$. Let $f: \mathbb{R}_{> 0} \to \mathbb{R}_{\geq 0}$ be a convex function such that $f(1) = 0$ and $f$ is strictly convex around 1. The $f$-divergence of $\nu$ from $\mu$ is defined by
\begin{align}
\dvg{\nu}{\mu}{f} = \E[X\sim \mu]{f\left(\frac{\nu(X)}{\mu(X)}\right)} = \sum_{\sigma \in \Omega}\mu(\sigma)f\left(\frac{\nu(\sigma)}{\mu(\sigma)} \right).
\end{align}
The $f$-divergence is a very general notion that measures the difference between two distributions. For instance, $f(x) = \frac{1}{2}|x-1|$ gives the \emph{total variation distance (TV-distance)}, $f(x) = \frac{1}{2}(x-1)^2$ gives the $\chi^2$-\emph{divergence}, and $f(x) = x \ln x - x + 1$ gives the \emph{Kullback-Leibler (KL) divergence}.

The \emph{Ising model} is a fundamental graphical model in statistical physics, probability theory, and machine learning.
Let $G=(V,E)$ be a graph. Let $J \in \mathbb{R}^{V \times V}$ be a \emph{symmetric interaction matrix} such that $J_{uv} \neq 0$ only if $\{u,v\} \in E$. Let $h \in \mathbb{R}^V$ be the \emph{external fields vector}. An Ising model specified by $(G,J,h)$ defines a \emph{Gibbs distribution} $\mu$ with support $\Omega = \{-1,+1\}^V$ such that 
\begin{align*}
    \forall \sigma \in \{-1,+1\}^V, \quad \mu(\sigma) = \frac{w_\mu(\sigma)}{Z_\mu} = \frac{\exp \left( \frac{1}{2}\sigma^T J \sigma + h^T \sigma\right)}{Z_\mu}, \quad\text{where } Z_\mu = \sum_{\sigma \in \{-1,+1\}^V} w_\mu(\sigma).
\end{align*}
The function $w_\mu: \Omega \to \mathbb{R}_{\geq 0}$ is called the \emph{weight function} of the Ising model and the normalization factor $Z_\mu$ is called the \emph{partition function} of the Ising model.

Recently, the problem of computing the TV-distance between two high-dimensional distributions has received increasing attention.
One interesting result~\cite{BGMMPV23} has proved that even for a pair of product distributions~(Ising models on an empty graph), the \emph{exact} computation of TV-distance is \#P-hard. 
Later on, polynomial time \emph{approximation algorithms} were proposed for product distributions~\cite{FGJW23,FengLL24}. 
For general Ising models, both algorithmic and hardness results were established~\cite{BGMMPV24ICLR,feng2025approximating} for approximating the TV-distance.

In this paper, we consider a more general problem of approximating the $f$-divergence between two Ising models. The problem is defined as follows.

\begin{problem}\label{prob:Ising}
Approximating the $f$-divergence for two Ising models. 
\begin{itemize}
    \item \textbf{Input}: Two Ising models $(G, J^\nu,h^\nu)$ and $(G,J^\mu,h^\mu)$ specifying two Gibbs distributions $\nu$ and $\mu$ respectively\footnote{\Cref{prob:Ising} assumes that $\nu$ and $\mu$ have the same underlying graph $G$. This assumption does not lose generality because the definition allows $J^\nu_{uv}$ and $J^\mu_{uv}$ to be 0 even if $\{u,v\} \in E$.}, a function $f$ defining the $f$-divergence, and an error bound parameter $\varepsilon$;
    \item \textbf{Output}: A number $\hat D \in \mathbb{R}$ such that $\e^{-\varepsilon}\dvg{\nu}{\mu}{f} \leq \hat D \leq \e^\varepsilon \dvg{\nu}{\mu}{f}$.
\end{itemize}
\end{problem}

A randomized algorithm is said to be an \emph{FPRAS} (fully polynomial randomized approximation scheme) for \Cref{prob:Ising} if it runs in time polynomial in $n = |V|$ and $1/\varepsilon$ and with probability at least $\frac{2}{3}$, the output approximates the value of the $f$-divergence within relative error $\e^{\pm \varepsilon}$.

Our algorithmic result is a reduction from $f$-divergence approximation to sampling and approximate counting, which are two fundamental computational tasks for the Ising model. 
There is long-line of research on developing efficient algorithms~\cite{levin2017markov} for sampling and approximate counting. We assume the following abstract oracles for sampling and approximate counting.

\begin{definition}[sampling and approximate counting oracles]\label{def:sample-approx-count-oracle}
    Let $(G,J^\mu,h^\mu)$ be an Ising model with Gibbs distribution $\mu$ and partition function $Z_\mu$. Let $\oracle^{\mathrm{sp}}_G,\oracle^{\mathrm{ct}}_G:\mathbb{R}_{>0} \rightarrow [|V|+|E|,\infty)$ be two non-increasing functions. Given any error bound $\varepsilon > 0$,
    \begin{itemize}
        \item The sampling oracle for $(G,J^\mu,h^\mu)$ with cost function $\oracle^{\mathrm{sp}}_G$ returns a random sample $X\in \{-1,+1\}^V$ in time $\oracle^{\mathrm{sp}}_G(\varepsilon)$ with $\DTV{X}{\mu} \leq \varepsilon$, where $\DTV{X}{\mu}$ is the total variation distance between $X$ and $\mu$.
        \item The approximate counting oracle for $(G,J^\mu,h^\mu)$ with cost function $\oracle^{\mathrm{ct}}_G$ returns a random number $\hat{Z}_\mu$ in time $\oracle^{\mathrm{ct}}_G(\varepsilon)$ with $\Pr{\e^{-\varepsilon}Z_\mu \leq \hat{Z}_\mu \leq \e^{\varepsilon}Z_\mu} \geq 0.99$.
    \end{itemize}
\end{definition}
For functions $\oracle^{\mathrm{sp}}_G(\varepsilon)$ and $\oracle^{\mathrm{ct}}_G(\varepsilon)$, we add the index $G$ to emphasize that the running time also depends on parameters of graph $G$ such as the number of vertices/edges and the maximum degree. We assume the oracles need to read the whole graph $G$ so that the cost is at least $|V| + |E|$.

We also require the following mild assumption on the Ising models.

\begin{definition}[marginal lower bound]\label{def:marginal-lower-bound}
  Let $b \geq 0$ be a constant. A Gibbs distribution $\mu$ is said to satisfy the $b$-marginal lower bound if for any $\Lambda \subseteq V$, any pinning $\sigma \in \{-1,+1\}^\Lambda$, any $v \in V \setminus \Lambda$, and any $c \in \{-1,+1\}$, it holds that $\mu^\sigma_v(c) \geq b$, where $\mu^\sigma_v(\cdot)$ denotes the marginal distribution on $v$ projected from $\mu$ conditional on that the configuration on $\Lambda$ is fixed as $\sigma$.
\end{definition}

The marginal lower bound condition is a natural and common assumption for the Ising model. The assumption is widely used in the literature of sampling and approximate counting~\cite{CLV21}, learning theory~\cite{Bresler15}, and TV-distance approximation~\cite{feng2025approximating}.

\subsection{Algorithm and hardness results for \texorpdfstring{$\chi^\alpha$}{chi-alpha}-divergence approximation}

Let $\alpha \geq 1$ be a constant integer. The $\dvg{\nu}{\mu}{f} = \dvg{\nu}{\mu}{\chi^\alpha}$ is called the $\chi^\alpha$-divergence if $f(x) = \frac{1}{2}|x-1|^\alpha$.
We give the following approximation algorithm for the $\chi^\alpha$-divergence between two Ising models in \Cref{prob:Ising}. Given two input Ising models $(G,J^\nu,h^\nu)$ and $(G,J^\mu,h^\mu)$, define the following family of Ising models on the graph $G$:
\begin{align}\label{eq:family-of-isings}
    \+F(\nu,\mu,\alpha) = \left\{ \text{Ising model } (G, J^{(k)}, h^{(k)}) \mid \substack{ \textstyle J^{(k)} \triangleq kJ^\nu - (k-1)J^\mu \\ \textstyle h^{(k)} \triangleq kh^\nu - (k-1)h^\mu} \text{ for integer } 0 \leq k \leq \alpha \right\}.
\end{align}
We remark that $(G,J^{(1)},h^{(1)})$ and $(G,J^{(0)},h^{(0)})$ are the same as the input Ising models $(G,J^\nu,h^\nu)$ and $(G,J^\mu,h^\mu)$ respectively. The following theorem says that if all Ising models in $\+F$ admit sampling and approximate counting oracles, then the $\chi^\alpha$-divergence between two input Ising models can be approximated in polynomial time.

\begin{theorem}\label{thm:approx-chi-alpha-divergence}
Let $\alpha \geq 1$ be a constant integer and $b \in (0,1)$ be a constant.
There exists an FPRAS that solves \Cref{prob:Ising} for $\chi^\alpha$-divergence in time 
$$O_{\alpha,b}\tp{\oracle^{\mathrm{ct}}_G\tp{\frac{\eta_{\alpha,b} \cdot \varepsilon}{(n+m)^\alpha}} + \frac{(n+m)^{2\alpha}}{\varepsilon^2} \cdot \oracle^{\mathrm{sp}}_G\left(\frac{\eta_{\alpha,b} \cdot \varepsilon^2}{(n+m)^{2\alpha}}\right)},$$ 
where $n = |V|$, $m = |E|$, and $\eta_{\alpha,b} > 0$ is a small constant depending only on $\alpha$ and $b$, if two input Ising Gibbs distributions $\mu$ and $\nu$ are both $b$-marginally bounded and all Ising models in $\+F(\nu,\mu,\alpha)$ admit sampling and approximate counting oracles with cost functions $\oracle^{\mathrm{sp}}_G(\cdot)$ and $\oracle^{\mathrm{ct}}_G(\cdot)$ respectively.
\end{theorem}

Let us consider a simplified case of the Ising model.
Suppose $G$ has the max degree $\Delta$.
For an Ising model $(G,J,h)$, if for any $e = \{u,v\} \in E$, if $J_{uv} = J_{vu} = \frac{\ln \beta}{2}$ take a unified value. In this case, an Ising model admits sampling and approximate counting oracles with cost function $\oracle^{\mathrm{sp}}_G(\varepsilon) = \mathrm{poly}(n, \log \frac{1}{\varepsilon})$ and $\oracle^{\mathrm{ct}}_G(\varepsilon) = \mathrm{poly}(n, \frac{1}{\varepsilon})$ if one of the following two conditions holds:
\begin{itemize}
    \item Uniqueness condition: $\frac{\Delta-2}{\Delta} \leq \beta \leq \frac{\Delta}{\Delta - 2}$ and an arbitrary external field $h \in \mathbb{R}^V$~\cite{CCYZ25}.
    \item Uniqueness condition or ferromagnetic condition with zero external field: $\beta \geq \frac{\Delta-2}{\Delta}$ and the zero external field $h = \boldsymbol{0}$~\cite{CCYZ25,JS93}.
\end{itemize}

Consider \Cref{prob:Ising} where two input Ising models both have unified values $\frac{\ln \beta_\nu}{2}$ and $\frac{\ln \beta_\mu}{2}$ in the interaction matrices $J^\nu$ and $J^\mu$. Assume $\beta_\nu$, $\beta_\mu$, $\Vert h^\nu \Vert_\infty$, $\Vert h^\mu \Vert_\infty$, and the max degree $\Delta$ are all constants. Then the marginal lower bound condition is satisfied. We have the following corollary.

\begin{corollary}\label{cor:approx-chi-alpha-divergence}
Let $\alpha \geq 1$ be a constant integer. For Ising models with unified values in interaction matrices, there exists an FPRAS that solves \Cref{prob:Ising} for $\chi^\alpha$-divergence in time $\mathrm{poly}(n, \frac{1}{\varepsilon})$ if $\beta_\nu$, $\beta_\mu$, $\Vert h^\nu \Vert_\infty$, $\Vert h^\mu \Vert_\infty$, and $\Delta$ are all constants and one of the following three conditions holds:
\begin{itemize}
    \item Both $\beta_\mu$ and $(\frac{\beta_\nu}{\beta_\mu})^{\alpha} \beta_\mu$ are in $[\frac{\Delta-2}{\Delta}, \frac{\Delta}{\Delta - 2}]$.
    \item $\beta_\nu \geq \beta_\mu \geq \frac{\Delta-2}{\Delta}$ and $h^\nu = h^\mu = \boldsymbol{0}$.
    \item $\beta_\nu < \beta_\mu$, $(\frac{\beta_\nu}{\beta_\mu})^{\alpha} \beta_\mu \geq \frac{\Delta-2}{\Delta}$, and $h^\nu = h^\mu = \boldsymbol{0}$.
\end{itemize}
\end{corollary}

The corollary requires different conditions for the case $\beta_\nu \geq \beta_\mu$ and the case $\beta_\nu < \beta_\mu$. It is reasonable because for $\alpha > 1$, the $\chi^\alpha$-divergence is \emph{not} symmetric $\dvg{\nu}{\mu}{\chi^\alpha} \neq \dvg{\mu}{\nu}{\chi^\alpha}$, and thus the roles of two input Ising models are different. 

In \Cref{thm:approx-chi-alpha-divergence}, we show that $\chi^\alpha$-divergence can be approximated if the sampling and approximate counting oracles for all Ising models in $\+F(\nu,\mu,\alpha)$ exist. In a special case when $\alpha = 1$, $\+F(\nu,\mu,1)$ only contains two input Ising models $(G,J^\nu,h^\nu)$ and $(G,J^\mu,h^\mu)$, where we recover the result in~\cite{feng2025approximating} for approximating the TV-distance. However, for general $\alpha > 1$, the family $\+F(\nu,\mu,\alpha)$ contains other Ising models. It is natural to ask the following questions.
\begin{tcolorbox}[
    title=Questions,
    colback=white,
    colframe=black,
    coltitle=black,
    colbacktitle=gray!20,
    fonttitle=\bfseries,
    boxrule=1pt,
    arc=3pt
]
\begin{itemize}[leftmargin=0.3cm]
    \item Is it possible to approximate $\chi^\alpha$-divergence if we only assume the sampling and approximate counting oracles for the two input Ising models $(G,J^\nu,h^\nu)$ and $(G,J^\mu,h^\mu)$ exist?
    \vspace{-0.18cm}
    \item A more direct question: are the oracle assumptions in \Cref{thm:approx-chi-alpha-divergence} necessary?
\end{itemize}
\end{tcolorbox}


The above two questions are answered by the following hardness result.
Let us restrict our attention to Ising models $(G,J,h)$ with zero external field and interaction matrix with unified values. Formally, $h = \boldsymbol{0}$ and $J_{uv} = J_{vu} = \frac{\ln \beta}{2}$ for all $\{u,v\} \in E$. Denote the model by $(G,\beta)$. Let $(G,\beta_\nu)$ and $(G,\beta_\mu)$ denote two input Ising models in \Cref{prob:Ising}.

\begin{theorem}\label{thm:hardness-without-F}
Fix an integer $\alpha \geq 2$, an integer $\Delta \geq 3$, and two parameters $\beta_\mu > \beta_\nu \geq \frac{\Delta-2}{\Delta}$ such that $(\frac{\beta_\nu}{\beta_\mu})^{\alpha} \beta_\mu < \frac{\Delta-2}{\Delta}$. 
Unless $\mathsf{NP} = \mathsf{RP}$, there is no FPRAS for $\chi^\alpha$-divergence for two Ising models $(G,\beta_\nu)$ and $(G,\beta_\mu)$ on $\Delta$-regular graph $G$.
\end{theorem}

The above theorem considers two Ising models with zero external field and $\beta_\nu,\beta_\mu \geq \frac{\Delta-2}{\Delta}$. Two input Ising models both admit polynomial time sampling and approximate counting oracles because they are either in the uniqueness regime $\beta \in [\frac{\Delta-2}{\Delta}, \frac{\Delta}{\Delta - 2}]$~\cite{CCYZ25} or the ferromagnetic regime $\beta \geq 1$~\cite{JS93}. However, if  $(\frac{\beta_\nu}{\beta_\mu})^{\alpha} \beta_\mu < \frac{\Delta-2}{\Delta}$, approximating $\chi^\alpha$-divergence is still hard. For a concrete example, the problem is hard when $\alpha \geq 2$, $\beta_\nu = \frac{\Delta-2}{\Delta}$, and $\beta_\mu = 2 \beta_\nu = \frac{2(\Delta - 2)}{\Delta}$. 

\Cref{thm:hardness-without-F} also shows that the condition in \Cref{thm:approx-chi-alpha-divergence} is necessary. 
Note that since $\Delta,\beta_\nu$ and $\beta_\mu$ are all constants and the external fields are zero, the hardness instances satisfy the marginal lower bound condition with $b = b(\Delta,\beta_\nu,\beta_\mu) = \Omega(1)$.
Consider two Ising models $(G,\beta_\nu)$ and $(G,\beta_\mu)$ on $\Delta$-regular graph $G$ such that $\beta_\mu > \beta_\nu \geq \frac{\Delta-2}{\Delta}$. For this family of instances, our algorithmic result in \Cref{cor:approx-chi-alpha-divergence} together with the hardness result in \Cref{thm:hardness-without-F} discovers the following \emph{computational phase transition} for approximating $\dvg{\nu}{\mu}{\chi^\alpha}$ when $\alpha \geq 2$:
\begin{itemize}
    \item FPRAS exists when $(\frac{\beta_\nu}{\beta_\mu})^{\alpha} \beta_\mu \geq \frac{\Delta-2}{\Delta}$;
    \item Unless $\mathsf{NP} = \mathsf{RP}$, there is no FPRAS when $(\frac{\beta_\nu}{\beta_\mu})^{\alpha} \beta_\mu < \frac{\Delta-2}{\Delta}$.
\end{itemize}

\Cref{thm:hardness-without-F} is proved via the hardness results of approximating partition functions of anti-ferromagnetic Ising models beyond the uniqueness regime~\cite{SS14,GSV16}. See \Cref{sec:proof-hardness} for details.

\subsection{Algorithms for other \texorpdfstring{$f$}{f}-divergences}

We say an Ising model $(G,J,h)$ admits polynomial time sampling and approximate counting oracles if it admits sampling and approximate counting oracles with cost function $\oracle^{\mathrm{sp}}_G(\varepsilon) = \mathrm{poly}(\frac{n}{\varepsilon})$ and $\oracle^{\mathrm{ct}}_G(\varepsilon) =\mathrm{poly}(\frac{n}{\varepsilon})$ respectively.
Let $(G,J^\nu,h^\nu)$ and $(G,J^\mu,h^\mu)$ denote two input Ising models. We give algorithms for approximating the $f$-divergence $\dvg{\nu}{\mu}{f}$ for various divergence functions $f$.
For many functions $f$, the divergence $\dvg{\nu}{\mu}{f}$ is \emph{not} symmetric for $\nu$ and $\mu$. We will refer to $(G,J^\nu,h^\nu)$ as the first input Ising model and $(G,J^\mu,h^\mu)$ as the second input Ising model.

By setting different functions $f$ in $\dvg{\nu}{\mu}{f}$, we can obtain the following divergences:
\begin{align*}
f(x) = \begin{cases}
x \ln x - x + 1, & \text{Kullback-Leibler divergence;} \\
-\ln x + x - 1, & \text{R\'enyi divergence;} \\
\frac{1}{2}\tp{x \ln x - (x+1)\ln \frac{x+1}{2}} & \text{Jensen-Shannon divergence.}
\end{cases}
\end{align*}


\begin{theorem}\label{thm:kl-divergence}
There exists an FPRAS for \Cref{prob:Ising} with KL-divergence, R\'enyi-divergence, and Jensen-Shannon divergence if two input Ising models are both marginally bounded and both admit polynomial time sampling and approximate counting oracles.
\end{theorem}

Compared with the result in \Cref{thm:approx-chi-alpha-divergence}, the above theorem only requires the sampling and approximate counting oracles for the two input Ising models exist.


Next, consider the $\alpha$-divergence.
When $\alpha = 0$, the $\alpha$-divergence refers to the KL-divergence. When $\alpha = 1$, the $\alpha$-divergence refers to the R\'enyi-divergence.
The $\alpha$-divergences for other values of $\alpha$ are defined as follows.
Let $\alpha \neq 0, 1$. The $\alpha$-divergence is defined as $f(x) = \frac{x^\alpha - \alpha x - (1-\alpha)}{\alpha(\alpha-1)}$. Recall that in \eqref{eq:family-of-isings}, $J^{(\alpha)} = \alpha J^\nu - (\alpha-1)J^\mu$ and $h^{(\alpha)} = \alpha h^\nu - (\alpha-1) h^\mu$. Here, $\alpha$ can be a \emph{real} number.

\begin{theorem}\label{thm:alpha-divergence}
    Let $\alpha \neq 0, 1$ be a constant real. 
    There exists an FPRAS for \Cref{prob:Ising} with $\alpha$-divergence if two input Ising models are both marginally bounded, the second input Ising model $(G, J^\mu, h^\mu)$ admits a polynomial time sampling oracle, and both two input Ising models together with $(G,J^{(\alpha)}, h^{(\alpha)})$ admit polynomial time approximate counting oracles.
\end{theorem}

The above theorem works for constant real $\alpha$ (assuming the computational model supports real arithmetic). Unlike $\chi^\alpha$-divergence, in addition to $\mu$ and $\nu$ we only require that the approximate counting oracle for $(G,J^{(\alpha)}, h^{(\alpha)})$ exists.

The squared Hellinger distance is defined by setting $f(x) = \frac{1}{2} (\sqrt{x} - 1)^2$. 
Given two input Ising models $(G,J^\nu,h^\nu)$ and $(G,J^\mu,h^\mu)$, define an averaged Ising model as $(G,J^\textsf{avg},h^\textsf{avg})$, where $J^\textsf{avg} = \frac{1}{2} (J^\nu + J^\mu)$ and $h^\textsf{avg} = \frac{1}{2} (h^\nu + h^\mu)$. We have the following theorem.

\begin{theorem}\label{thm:squared-hellinger-distance}
    There exists an FPRAS for \Cref{prob:Ising} with squared Hellinger distance if two input Ising models are both marginally bounded, one of the input Ising models admits a polynomial time sampling oracle, and two input Ising models together with the averaged Ising model all admit polynomial time approximate counting oracles.
\end{theorem}

The squared Hellinger distance is symmetric $\dvg{\nu}{\mu}{f} = \dvg{\mu}{\nu}{f}$. Hence, we only require one of the input Ising models admits a polynomial time sampling oracle.
As a corollary, if two input Ising models are both in the uniqueness regime, then their averaged Ising model is also in the uniqueness regime and FPRAS exists for the squared Hellinger distance. For two Ising models $(G,\beta_\nu)$ and $(G,\beta_\mu)$ with zero external field, FPRAS exists when both $\beta_\nu,\beta_\mu \geq \frac{\Delta-2}{\Delta}$.

\subsection{Related work and open problems}

\paragraph{Related work}
The problem of computing the $f$-divergence between two (either discrete or continuous) distributions is well-motivated by the applications in machine learning and statistics. There are many works, both theoretical and practical, studying the algorithms in various settings. Many works study the error between the true divergence and the divergence computed by certain empirical distributions, e.g.,~\cite{BJPS11,KanamoriSS12, RubensteinBDRT19,SreekumarG22}.
Embedding-base technique was also studied in~\cite{AbdullahKMVV16}, where algorithm embeds a large sample space into a smaller one while preserving the $f$-divergence.
We cannot directly use these techniques to Ising model as the size of the sample space is $2^n$, which gives an exponential time algorithm.
For graphical models, \cite{LeePPW23}~shows that certain divergences can be computed exactly when the graphical model has a \emph{bounded treewidth}.

Approximating the $f$-divergence is related to the testing problem, in which one or both two distributions are given by the access to different types of sampling oracles, and the algorithm is required to test certain divergence is large or small. There is a long line of work studying the testing problem. See \cite{Canonne15,Canonne22} for a comprehensive survey.

The TV-distance $\DTV{\cdot}{\cdot}$ is a special case of $f$-divergences. There are many theoretical works studying the approximation algorithms (either additive error or relative error) and hardness of approximation for the TV-distance between various types of high-dimensional distributions. For example, distributions specified by circuits~\cite{sahai2003complete}, hidden Markov models~\cite{Kiefer18}, product distributions~\cite{BGMMPV23,FGJW23,FengLL24,kontorovich2025tensorization}, graphical models~\cite{BGMV20,BGMMPV24ICML,feng2025approximating}, and high-dimensional Gaussian distributions~\cite{BGMV20,BFS25}. Recently, there are some works focusing on approximating $1 - \DTV{\cdot}{\cdot}$ for two product distributions ~\cite{kontorovich2024sharp,bhattacharyya2025algorithms}.

\paragraph{Open problems}
A natural direction is to consider more general distributions and more general divergences. 
Here are a few examples.
For the Ising model, how to remove the marginal lower bound assumption? There are many other types of graphical models, such as Markov random fields with high-order interactions and Bayesian networks.
It is interesting to consider other types of distributions, e.g., distributions generated by probabilistic circuits~\cite{amarilli2024circus}.
For the divergence, we focus on $\chi^\alpha$-divergence when $\alpha$ is an integer. One can also consider the case when $\alpha$ is a real number and give an algorithm that works in the tight parameter regime.
Another question is to give a more general algorithm that works for a larger family of $f$-divergences.

The algorithm given in this paper is randomized. An open question is to find deterministic algorithms by exploring the deterministic approximate counting algorithms for the Ising model~\cite{LLY13,LSS19}. Currently, deterministic algorithms are known for approximating the TV-distance between two product distributions~\cite{BGMMPV23,FengLL24}.

We can also study the same problem in different settings.
In our setting, we assume that the input gives the description of two Ising models. It is interesting to consider another (perhaps more practical) setting that the algorithm can only access one or two models through random samples. 
This setting is closely related to the \emph{Ising testing problem} in~\cite{DaskalakisDK18}. 
One can even consider abstract distributions with certain properties (e.g. approximate tensorization of $f$-divergence) and the distribution is given by the access to some oracles (e.g., oracle for querying conditional marginal distributions or querying probability masses).  This abstract setting was considered by recent work in testing~\cite{BlancaCSV25,gay2025sampling}.

\section{Algorithm overview}\label{sec:overview}

We give an overview of our algorithm for approximating the $\chi^\alpha$-divergence between two Ising models.
We start with the following definition of the \emph{parameter distance} between two Ising models.

\begin{definition}[\text{parameter distance~\cite{feng2025approximating}}]
    For two Ising models $(G,J^\nu,h^\nu)$ and $(G,J^\mu,h^\mu)$, the parameter distance $d_{\mathrm{par}}(\nu,\mu)$ is defined by
    \begin{align*}
        d_{\mathrm{par}}(\nu,\mu) \triangleq \max\left\{\left\lVert J^\nu-J^\mu\right\rVert _{\max},\max_v{\frac{\left| h^\nu(v)-h^\mu(v)\right|}{\deg (v) +1}}\right\},
    \end{align*}
    where $\deg(v)$ is the degree of $v$ in $G$.
\end{definition}

In~\cite{feng2025approximating}, the parameter distance is used to give a lower bound of TV-distance $\DTV{\nu}{\mu}$, with which the previous work gave an FPRAS for TV-distance.
The following lemma says such a lower bound can be generalized to $\chi^\alpha$-divergence.
The lemma is proved in~\Cref{sec:parameter-distance-f-divergence}, where the proof works for a general family of $f$-divergence.

\begin{lemma}[\text{$\chi^\alpha$-divergence lower bound}]\label{lem:chi-alpha-lower-bound}
    For $f = \frac{1}{2} |x-1|^\alpha$, where $\alpha \geq 1$, it holds that
    \begin{align*}
        \dvg{\nu}{\mu}{\chi^\alpha} \geq \frac{b^{2\alpha}}{2} \cdot d_{\mathrm{par}}(\nu,\mu)^\alpha.
    \end{align*}
\end{lemma}

 Let $\theta = \frac{1}{\mathrm{poly}(n)}$ be a threshold. We estimate $\dvg{\nu}{\mu}{\chi^\alpha}$ in the following two cases.
\begin{itemize}
    \item \textbf{Case} $d_{\mathrm{par}}(\nu,\mu) > \theta$. This case is trivial for TV-distance because there is a very simple additive-error approximation algorithm for the TV-distance~\cite{BGMV20}. Since TV-distance itself is lower bounded by ${1}/{\mathrm{poly}(n)}$, the additive error approximation can be easily transferred to relative error approximation. However, this simple algorithm only works for the TV-distance. Instead, we propose a new algorithm for approximating the general $\chi^\alpha$-divergence. The algorithm is outlined in \Cref{sec:hard-parameter-distance}.

    \item \textbf{Case} $d_{\mathrm{par}}(\nu,\mu) \leq \theta$. In this case, two Ising models are similar to each other. Previous work \cite{feng2025approximating} has explored the similarity of two models and designed an algorithm for approximating their TV-distance. 
    We substantially generalize their algorithm to approximate a family  $f$-divergence (including $\chi^\alpha$-divergence). The algorithm is outlined in~\Cref{sec:small-parameter-distance}.
\end{itemize}

\subsection{The algorithm for instances with large parameter distance}\label{sec:hard-parameter-distance}

We first state the challenge of approximating the $\chi^\alpha$-divergence for general $\alpha$ compared to the TV-distance.
It is well-known that the TV-distance can be written as follows:
\begin{align}\label{eq:tv-distance-expectation}
    \DTV{\nu}{\mu} = \frac{1}{2} \E[X \sim \mu]{\left| \frac{\nu(X)}{\mu(X)} - 1 \right|} = \frac{1}{2}\sum_{\sigma \in \Omega}|\mu(\sigma)-\nu(\sigma)| \overset{(\star)}{=} \sum_{\sigma: \nu(\sigma) < \mu(\sigma)}\mu(\sigma)\tp{1 - \frac{\nu(\sigma)}{\mu(\sigma)}}.
\end{align}
It was observed in \cite{BGMV20} that the TV-distance can be approximated by draw independent samples $X \sim \mu$ and then taking the average of $Y = \max\{0,1-\frac{\nu(X)}{\mu(X)}\}$\footnote{The value of $Y$ can be computed approximately, which is sufficient for the purpose of approximation.}. The above equation shows $\E[]{Y} = \DTV{\nu}{\mu}$. It is easy to see $|Y| \leq 1$ and thus $\Var[]{Y} \leq 1$. The simple algorithm achieves the additive error approximation, which can be transferred to relative-error approximation because $\DTV{\nu}{\mu} \geq 1/\mathrm{poly}(n)$ is lower bounded in this case.

However, for the $\chi^\alpha$-divergence with general $\alpha$, by the definition,
\begin{align*}
    \dvg{\nu}{\mu}{\chi^\alpha} = \frac{1}{2} \sum\limits_{\sigma \in \Omega} \mu(\sigma) \left| \frac{\nu(\sigma)}{\mu(\sigma)}-1 \right|^\alpha = \frac{1}{2} \sum_{\sigma \in \Omega}   |\mu(\sigma)-\nu(\sigma)| \cdot \left \vert \frac{\nu(\sigma)}{\mu(\sigma)} - 1 \right \vert ^{\alpha - 1}. 
\end{align*}
Due to the extra term $\vert \frac{\nu(\sigma)}{\mu(\sigma)} - 1  \vert ^{\alpha - 1}$, we cannot apply a similar transformation as equation $(\star)$ in~\eqref{eq:tv-distance-expectation}. Hence, it is not clear how to design an unbiased estimator $Y$ such that variance of $Y$ is small.

To overcome this challenge, we propose a new algorithm for approximating the $\chi^\alpha$-divergence. The following lemma gives a lower bound for the $\chi^\alpha$-divergence in this case.


\begin{lemma}\label{lem:error-bound}
    Let $0<\theta,b \leq 1$. If $d_{\mathrm{par}}(\nu,\mu) \geq \theta$ and both $\nu$ and $\mu$ are $b$-marginally bounded, then
    \begin{align*}
        \dvg{\nu}{\mu}{\chi^\alpha} \geq  B_{\alpha,b}(\theta) \cdot \sum\limits_{\sigma \in \{\pm\}^V} \mu(\sigma) \left( \frac{\nu(\sigma)}{\mu(\sigma)} + 1 \right)^\alpha.
    \end{align*}
    where $B_{\alpha,b}(\theta) =  \frac{b^{2\alpha} \theta^\alpha}{2} \left( \left( \frac{b^{2\alpha} \theta^\alpha}{2} \right)^{1/(\alpha+1)} + 2 \right)^{-\alpha} \left( 2 \left( \frac{b^{2\alpha} \theta^\alpha}{2} \right)^{1/(\alpha+1)} + 1 \right)^{-1} = \Theta_{\alpha,b}(\theta^\alpha)$.
\end{lemma}

The proof of \Cref{lem:error-bound} separates all $\sigma \in \Omega$ into two parts, and separately lower bound their contributions to the $\chi^\alpha$-divergence. The detailed proof is deferred to~\Cref{sec:lower-bound-large-parameter-distance}. 

\Cref{lem:error-bound} gives us a tool to control the error of approximation. Recall that our goal is to design an algorithm that approximates the $\chi^\alpha$-divergence with relative error $\e^{\pm\varepsilon}$. Equivalently, the algorithm should achieve the $O(\varepsilon) \cdot \dvg{\nu}{\mu}{\chi^\alpha}$ additive error. Using the above lemma, it is sufficient to show that the algorithm can achieve the $O(\varepsilon) \cdot B_{\alpha,b}(\theta) \cdot \sum_{\sigma \in \Omega} \mu(\sigma) \left( \frac{\nu(\sigma)}{\mu(\sigma)} + 1 \right)^\alpha$ additive error approximation, which turns out to be much easier to prove.  In addition to algorithms, \Cref{lem:error-bound} also plays an important role in the proof of the hardness result. See \Cref{sec:proof-hardness} for details. 



\subsubsection{The algorithm for even \texorpdfstring{$\alpha$}{alpha}}\label{sec:overview-algorithm-even-alpha}

Our actual algorithm deals with all $\alpha \geq 1$ in a unified way. However, in the overview, we exhibit a simple algorithm for the case where $\alpha$ is even. The simple case will illustrate the intuition why we need to consider a family of Ising models $\+F(\nu,\mu,\alpha)$ in~\eqref{eq:family-of-isings} and why \Cref{lem:error-bound} can help us to control the error of approximation.
When $\alpha$ is even, we can replace $\left \vert \frac{\nu(\sigma)}{\mu(\sigma)} - 1 \right \vert ^{\alpha}$ with $\left( \frac{\nu(\sigma)}{\mu(\sigma)} - 1 \right)^{\alpha}$ and then use the binomial expansion to get the following equation:
\begin{align}
    \dvg{\nu}{\mu}{\chi^\alpha} = \frac{1}{2} \sum\limits_{\sigma \in \Omega} \mu(\sigma) \left( \frac{\nu(\sigma)}{\mu(\sigma)}-1 \right)^\alpha
    = \frac{1}{2} \sum\limits_{k=0}^\alpha \binom{\alpha}{k} (-1)^{\alpha - k} \sum\limits_{\sigma \in \Omega} \frac{\nu^k(\sigma)}{\mu^{k-1}(\sigma)}.\label{eq:expansion-for-even-alpha}
\end{align}

To deal with the term $\frac{\nu^k(\sigma)}{\mu^{k-1}(\sigma)}$, recall that~\eqref{eq:family-of-isings} defines a family of Ising models $(G, J^{(k)}, h^{(k)})$ such that $J^{(k)} \triangleq kJ^\nu - (k-1)J^\mu$ and $h^{(k)} \triangleq kh^\nu - (k-1)h^\mu$.
By our assumption in \Cref{thm:approx-chi-alpha-divergence}, for all $0 \leq k \leq \alpha$, the Ising model $(G, J^{(k)}, h^{(k)})$ admits sampling and approximate counting oracles. Let $Z_k$ denote the \emph{partition function} of the Ising model $(G, J^{(k)}, h^{(k)})$. We remark that $Z_\nu = Z_1$ and $Z_\mu = Z_0$ are the partition functions of input Ising model $(G, J^\nu, h^\nu)$ and $(G, J^\mu, h^\mu)$, respectively. Each term in the summation of~\eqref{eq:expansion-for-even-alpha} can be written as
\begin{equation}\label{eq:k-ising-model}
    \sum\limits_{\sigma \in \Omega} \frac{\nu^k(\sigma)}{\mu^{k-1}(\sigma)} = \frac{Z_\mu^{k-1} }{Z_\nu^k} \sum_{\sigma \in \Omega} \frac{w_\nu^k(\sigma)}{w_\mu^{k-1}(\sigma)} = \frac{Z_\mu^{k-1} }{Z_\nu^k} \sum_{\sigma \in \Omega} \exp\left(\frac{\sigma^T J^{(k)} \sigma}{2} + \sigma^T h^{(k)}\right) = \frac{Z_\mu^{k-1} \cdot Z_k }{Z_\nu^k}.
\end{equation}

The above ratio can be approximated by using approximate counting oracles for $Z_\nu$, $Z_\mu$, and $Z_k$. Note that $0 \leq k \leq \alpha$ is a constant. By choosing the relative approximation error $O(B_{\alpha,b}(\theta) \cdot \varepsilon) = \mathrm{poly}(\varepsilon/n)$ in the oracle small enough, where $B_{\alpha,b}(\theta) = \mathrm{poly}(1/n)$ is the parameter in \Cref{lem:error-bound}, we can obtain a value $W_k \in \e^{\pm O(B_{\alpha,b}(\theta) \cdot \varepsilon)} \cdot \frac{Z_\mu^{k-1} \cdot Z_k }{Z_\nu^k}$ that achieves the following approximation guarantee
\begin{align}\label{eq:k-ising-model-approximation}
\left \vert W_k  - \frac{Z_\mu^{k-1} \cdot Z_k }{Z_\nu^k} \right \vert \leq B_{\alpha,b}(\theta) \cdot \varepsilon \cdot  \frac{Z_\mu^{k-1} \cdot Z_k }{Z_\nu^k},
\end{align}
where in the above inequality, we write the relative error in terms of the additive error. Finally, our algorithm outputs the number $\hat{D}$ such that
\begin{align*}
    \hat{D} = \frac{1}{2} \sum\limits_{k=0}^\alpha \binom{\alpha}{k} (-1)^{\alpha - k}W_k.
\end{align*}
Note that $\hat{d}$ is an interlacing of plus and minus. Using~\eqref{eq:expansion-for-even-alpha} and~\eqref{eq:k-ising-model-approximation}, in the worst case (the additive error of every term takes the worst sign), the error of $\hat{d}$ can be upper bounded as follows
\begin{align*}
    \left| \hat{D} - \dvg{\nu}{\mu}{\chi^\alpha} \right| &\leq \varepsilon \cdot B_{\alpha,b}(\theta) \cdot \sum\limits_{k=0}^\alpha \binom{\alpha}{k} \frac{Z_\mu^{k-1} \cdot Z_k }{Z_\nu^k} = \varepsilon \cdot B_{\alpha,b}(\theta) \cdot \sum\limits_{k=0}^\alpha \binom{\alpha}{k} \sum\limits_{\sigma \in \Omega} \frac{\nu^k(\sigma)}{\mu^{k-1}(\sigma)}\notag\\
    &= \varepsilon \cdot B_{\alpha,b}(\theta) \cdot \sum\limits_{\sigma \in \Omega} \mu(\sigma) \left( \frac{\nu(\sigma)}{\mu(\sigma)} + 1 \right)^\alpha \overset{\text{\Cref{lem:error-bound}}}{\leq} \varepsilon \cdot \dvg{\nu}{\mu}{\chi^\alpha}.
\end{align*}

\subsubsection{The algorithm for general \texorpdfstring{$\alpha$}{alpha}}

For general $\alpha$, we need to consider the effect of the absolute value inside the divergence.
The $\chi^\alpha$-divergence
$\dvg{\nu}{\mu}{\chi^\alpha}
        = \frac{1}{2}\sum_{\sigma\in\Omega} \mu(\sigma) |\frac{\nu(\sigma)}{\mu(\sigma)} - 1|^\alpha$ can be written as
    \begin{align*}
        \dvg{\nu}{\mu}{\chi^\alpha}
        = \frac{1}{2}\sum_{\sigma: \nu(\sigma) > \mu(\sigma)}\mu(\sigma) \tp{\frac{\nu(\sigma)}{\mu(\sigma)} - 1}^\alpha - \frac{1}{2}\sum_{\sigma: \nu(\sigma) < \mu(\sigma)} \mu(\sigma)\tp{\frac{\nu(\sigma)}{\mu(\sigma)} - 1}^\alpha.
    \end{align*}
Using the binomial expansion, the divergence can be written as
\begin{align*}
    \dvg{\nu}{\mu}{\chi^\alpha} = \frac{1}{2}\sum_{k=0}^{\alpha} (-1)^{\alpha - k} \binom{\alpha}{k} \left( \sum\limits_{\sigma: \nu(\sigma) > \mu(\sigma)} \frac{\nu^k(\sigma)}{\mu^{k-1}(\sigma)} + (-1)^\alpha \sum\limits_{\sigma: \nu(\sigma) < \mu(\sigma)} \frac{\nu^k(\sigma)}{\mu^{k-1}(\sigma)} \right).
\end{align*}

Fix an integer $0 \leq k \leq \alpha$.
Let $X$ be a random sample from the Ising model $(G, J^{(k)}, h^{(k)})$. It holds that $\Pr[]{X = \sigma} \propto \frac{\nu^k(\sigma)}{\mu^{k-1}(\sigma)}$.
Define $W^+_k$ and $W^-_k$ to be random variables such that 
\begin{align}\label{eq:definition-for-W-plus-minus}
    W^+_k = \mathbf{1}[\nu(X) > \mu(X)] \cdot \frac{Z_0^{k-1} \cdot Z_k}{Z_1^k} \quad \text{and} \quad 
    W^-_k = \mathbf{1}[\nu(X) < \mu(X)] \cdot \frac{Z_0^{k-1} \cdot Z_k}{Z_1^k}.
\end{align}
With some calculation, we can verify that
\begin{align*}
    \sum\limits_{\sigma: \nu(\sigma) > \mu(\sigma)} \frac{\nu^k(\sigma)}{\mu^{k-1}(\sigma)} + (-1)^\alpha \sum\limits_{\sigma: \nu(\sigma) < \mu(\sigma)} \frac{\nu^k(\sigma)}{\mu^{k-1}(\sigma)} = \E{W^+_k} + (-1)^\alpha \E{W^-_k}.
\end{align*}

Ideally, our algorithm wants to draw $T = \mathrm{poly}(\frac{\log n}{\varepsilon})$ independent random samples of $W^+_k$ and then uses their average value $\hat{W}^+_k$ to approximate $\mathbf{E}[W^+_k]$.
Similarly, the algorithm computes $\hat{W}^-_k$ to approximate $\mathbf{E}[W^-_k]$. Finally, it outputs $\hat D = \frac{1}{2} \sum_{k=0}^{\alpha} (-1)^{\alpha - k} \binom{\alpha}{k} (\hat{W}^+_k + (-1)^\alpha \hat{W}^-_k)$.
However, to show the correctness of the algorithm, we need to deal with the following two types of errors.
\begin{itemize}
    \item The algorithm cannot draw perfect samples of $W^+_k$ or $W^-_k$. One major issue\footnote{There are some other issues, e.g., the samples from the Ising model $(G, J^{(k)}, h^{(k)})$ are approximate and partition functions $Z_0,Z_1,Z_k$ in~\eqref{eq:definition-for-W-plus-minus} can only be approximated.} is that given a sample $X \sim \text{Ising}(G, J^{(k)}, h^{(k)})$, the exact computation of the probability masses $\nu(X)$ and $\mu(X)$ is $\#$P-hard~\cite{JS93}. Hence, the algorithm cannot perfectly distinguish the cases where $\nu(X) > \mu(X)$ or $\nu(X) < \mu(X)$. This issue can be solved by using the approximate counting oracle to approximate the probability masses $\nu(X)$ and $\mu(X)$. By choosing a small enough error parameter $\mathrm{poly}(n/\varepsilon)$ for approximate counting, the algorithm misclassifies $\nu(X) > \mu(X)$ or $\nu(X) < \mu(X)$ only if $|\frac{\nu(X)}{\mu(X)} - 1|  \leq \mathrm{poly}(\varepsilon/n)$. We need to show that such $X$ does not contribute much to the error of approximation.
    \item We use the average of $\hat{W}^+_k$ and $\hat{W}^-_k$ to approximate $\mathbf{E}[W^+_k]$ and $\mathbf{E}[W^-_k]$ respectively. We need to control the concentration error. This error can be bounded via \Cref{lem:error-bound}, using a similar technique as described in the even $\alpha$ case.
\end{itemize}
The detailed algorithm and analysis is in \Cref{sec:algorithm-odd-alpha}.


\subsection{The algorithm for instances with small parameter distance}\label{sec:small-parameter-distance}


In this case, we give a general algorithm for abstract $f$-divergence.
We briefly sketch the algorithm implemented for $\chi^\alpha$-divergence in the overview.
The abstract one is in \Cref{sec:alg-abs}.
By the definition,
\begin{align}\label{eq:expansion-for-small-parameter-distance}
    \dvg{\nu}{\mu}{\chi^\alpha} = \frac{1}{2} \sum\limits_{\sigma \in \{\pm\}^V} \mu(\sigma) \left| \frac{\nu(\sigma)}{\mu(\sigma)}-1 \right|^\alpha = \frac{1}{2} \sum\limits_{\sigma \in \{\pm\}^V} \mu(\sigma) \left| \frac{w_\nu(\sigma)}{w_\mu(\sigma)} \cdot \frac{Z_\mu}{Z_\nu} -1 \right|^\alpha.
\end{align}

Define a random variable $W \triangleq \frac{w_\nu(X)}{w_\mu(X)}$, where $X \sim \mu$.
It can be verified that $\E{W} = \frac{Z_\nu}{Z_\mu}$ and 
\begin{align*}
    \dvg{\nu}{\mu}{\chi^\alpha}  = \E{\left\vert \frac{W}{\E[]{W}} - 1 \right\vert^\alpha}.
\end{align*}
The algorithm draws $T = \mathrm{poly}(n/\varepsilon)$ samples $W_1, \ldots, W_T$ from $W$ independently, and then computes $\bar{W} = \frac{1}{T}\sum_{i=1}^T W_i$ and $\hat{D} = \frac{1}{T}\sum_{i=1}^T \frac{1}{2} \left| \frac{W_i}{\bar W} -1 \right|^\alpha$. 
Since the parameter distance is small, two weight functions $w_\nu(\cdot)$ and $w_\mu(\cdot)$ are very similar. Hence, the ratio $W = \frac{w_\nu(X)}{w_\mu(X)} \approx 1$. Formally,
\begin{align*}
 \forall \sigma \in \Omega, \quad \left| \frac{w_\nu(\sigma)}{w_\mu(\sigma)} - \E{W} \right|^\alpha \leq \mathrm{poly}(n) \cdot d_{\mathrm{par}}(\nu,\mu)^\alpha \overset{\text{\Cref{lem:chi-alpha-lower-bound}}}{\leq}\mathrm{poly}(n) \cdot \dvg{\nu}{\mu}{\chi^\alpha}.
\end{align*}
The above property guarantees that $W$ is well-concentrated around its mean. By choosing the number of samples $T = \mathrm{poly}(n/\varepsilon)$ large enough, with some calculation, we can show that $\hat{D}$ approximates $\dvg{\nu}{\mu}{\chi^\alpha}$ with a relative error of $\e^{\pm O(\varepsilon)}$ with high probability. We remark that $\hat{D}$ is not necessarily an unbiased estimator for $\dvg{\nu}{\mu}{\chi^\alpha}$ but it is still a good approximation.

\subsection{Organization of the paper}
In \Cref{sec:lower-bound}, we prove some lower bounds for a broad family of $f$-divergences.
In \Cref{sec:alg-abs}, we give a general algorithm for general $f$-divergences satisfying an abstract condition.
In \Cref{sec:alg-chi-alpha}, we give the detailed algorithm for $\chi^\alpha$-divergence, where we focus on the large parameter distance case because the small parameter distance case is solved in \Cref{sec:alg-abs}.
In \Cref{sec:other-f-divergences}, we show how to extend our algorithm to other $f$-divergences. Finally, in \Cref{sec:proof-hardness}, we prove the hardness result for approximating $\chi^\alpha$-divergence in the parameter regime stated in \Cref{thm:hardness-without-F}.

\section{Lower bounds of \texorpdfstring{$f$}{f}-divergences}\label{sec:lower-bound}

In this paper, we assume the function $f$ in $f$-divergence satisfies the following assumption.
\begin{assumption}\label{ass:f-assumption}
    The convex continuous function $f:(0,+\infty) \to [0,\infty)$ satisfies that $f(1) = 0$ and 
    for any $x \neq 1$, $f$ is twice differentiable at $x$ such that $f''(x) \geq 0$. Furthermore, the derivative satisfies $f'(x) < 0$ if $0 < x < 1$ and $f'(x) > 0$ if $x > 1$.
\end{assumption}
For any function $f$ satisfying the above assumption, $x = 1$ is the unique point where $f(x) = 0$. It is straightforward to verify that $f$ is strictly convex at around 1. The $f$-divergences we are interested in all satisfy the above assumption.

\subsection{Lower bound in terms of parameter distance}\label{sec:parameter-distance-f-divergence}
Our proof uses the following lower bound of total variation distance.

\begin{lemma}[\text{\cite[Lemma 15]{feng2025approximating}}]\label{lem:tv-lower-bound}
  For two Ising models,  if both $\nu$ and $\mu$ are $b$-marginally bounded, then their total variation distance is lower bounded by
    \begin{align*}
        \DTV{\nu}{\mu} \geq \frac{b^2}{2} \cdot d_{\mathrm{par}}(\nu,\mu).
    \end{align*}
\end{lemma}


\begin{lemma}\label{lem:f-lower-bound}
    Suppose $f$ satisfies \Cref{ass:f-assumption}.
    For two Gibbs distributions $\nu,\mu$ of Ising models,
    \begin{align*}
        \dvg{\nu}{\mu}{f} &\geq \max\{ f(1-\DTV{\nu}{\mu}), f(1+\DTV{\nu}{\mu}) \}\\
        &\geq \max \left\{f\tp{1-\frac{b^2}{2}d_{\mathrm{par}}(\nu,\mu)}, f\tp{1+\frac{b^2}{2}d_{\mathrm{par}}(\nu,\mu)} \right\}.
    \end{align*}
\end{lemma}

\begin{proof}[Proof of \Cref{lem:chi-alpha-lower-bound}]
    Let $A \triangleq \{x \in \Omega: \nu(x) > \mu(x)\}$.
    Let $p = \Pr[X \sim \nu]{X \in A}$ and $q = \Pr[X \sim \mu]{X \in A}$.
    It is well-known that the total variation distance between $\nu$ and $\mu$ is $\DTV{\nu}{\mu} = p - q$.
    Since we consider the Ising model, $p,q > 0$.
    Consider a Markov kernel $K: \Omega \to \{0,1\}$ such that given any $\sigma \in \Omega$, $K$ deterministically transforms $\sigma$ to $1$ if and only if $\sigma \in A$. Note that $\nu K$ and $\mu K$ are Bernoulli distributions with parameters $p$ and $q$ respectively. We have
    \begin{align*}
        \DTV{\nu}{\mu} = \DTV{\nu K}{\mu K} = p - q.
    \end{align*}
    By using the data processing inequality for $f$-divergence, we have
    \begin{align*}
        \dvg{\nu}{\mu}{f} &\geq \dvg{\nu K}{\mu K}{f}\\
        &= q \cdot f \tp{\frac{p}{q}} + (1-q) \cdot f\tp{\frac{1-p}{1-q}}\\
        \text{$\left(f\left(\frac{1-p}{1-q}\right) \geq f(1) = 0\right)$} \quad &\geq q \cdot f \tp{\frac{p}{q}} + (1-q) \cdot f(1)\\
        \text{(Jensen's inequality on $f$)} \quad &\geq f\tp{1+p-q} = f\tp{1+\DTV{\nu}{\mu}}.
    \end{align*}
    Similarly, we use $f(p/q) \geq f(1)$ to get $\dvg{\nu}{\mu}{f} \geq f\tp{1-\DTV{\nu}{\mu}}$. 

    Note that $f'(x) < 0$ if $0 < x < 1$ and $f'(x) > 0$ if $x > 1$.
    Combining with \Cref{lem:tv-lower-bound}, we have 
    \begin{align*}
        f(1-\DTV{\nu}{\mu}) \geq f\tp{1-\frac{b^2}{2}d_{\mathrm{par}}(\nu,\mu)} \text{ and }
        f(1+\DTV{\nu}{\mu}) \geq f\tp{1+\frac{b^2}{2}d_{\mathrm{par}}(\nu,\mu)}. &\qedhere
    \end{align*}
\end{proof}

\Cref{lem:chi-alpha-lower-bound} for $\chi^\alpha$-divergence can be proved as follows.

\begin{proof}[Proof of \Cref{lem:chi-alpha-lower-bound}]
For $\chi^\alpha$-divergence, we have $f(x) = \frac{1}{2}|x-1|^\alpha$. We can get a slightly better lower bound than the general result in \Cref{lem:f-lower-bound}. Note that $f(1 - x) = f(1 + x)$. Then $f( \frac{1-p}{1-q}) = f( \frac{1+p-2q}{1-q})$. Using the Jensen's inequality on $q \cdot f \tp{\frac{p}{q}} + (1-q) \cdot f( \frac{1+p-2q}{1-q} )$ implies the a lower bound $\dvg{\nu}{\mu}{\chi^\alpha} \geq f(1+2\DTV{\nu}{\mu}) =  2^{\alpha -1} \DTV{\nu}{\mu}^\alpha$. Then lemma follows from \Cref{lem:tv-lower-bound}.
\end{proof}

\subsection{A lower bound for \texorpdfstring{$\chi^\alpha$}{chi-alpha}-divergence}\label{sec:lower-bound-large-parameter-distance}

\begin{proof}[Proof of \Cref{lem:error-bound}]
    Let $0 \leq t < 1$ be a parameter to be fixed later.
    We partition the whole space $\Omega = \{\pm\}^V$ into $M$ and $\overline{M}= \Omega \setminus M$ such that 
    \begin{align*}
        M = \left\{ \sigma \mid \left| \frac{\nu(\sigma)}{\mu(\sigma)} -1 \right| \leq t \right\}.
    \end{align*}
    Then, we bound the contribution of $\sigma \in M$ and $\sigma \in \overline{M}$ separately. By triangle inequality, we have $\frac{\nu(\sigma)}{\mu(\sigma)} + 1 \leq 2 + \left| \frac{\nu(\sigma)}{\mu(\sigma)} - 1 \right|$.
    For the first case, we have 
    \begin{align*}
        \sum_{\sigma \in M} \mu(\sigma) \left( \frac{\nu(\sigma)}{\mu(\sigma)} + 1 \right)^\alpha &\leq \sum_{\sigma \in M} \mu(\sigma) \left( 2 + t \right)^\alpha \leq \left( 2 + t \right)^\alpha.
    \end{align*}
    By our assumption, $d_{\mathrm{par}}(\nu,\mu) \geq \theta$. Using \Cref{lem:chi-alpha-lower-bound}, we have $\dvg{\nu}{\mu}{\chi^\alpha} \geq \frac{b^{2\alpha}}{2} \theta^\alpha$. Hence,
    \begin{align}\label{eq:M}
        \sum_{\sigma \in M} \mu(\sigma) \left( \frac{\nu(\sigma)}{\mu(\sigma)} + 1 \right)^\alpha \leq \frac{2 \cdot \left( 2 + t \right)^\alpha}{b^{2\alpha}\theta^\alpha} \dvg{\nu}{\mu}{\chi^\alpha}.
    \end{align}
    Consider the other set $\overline{M}$. For all $\sigma \in \overline{M}$, it holds that
    \begin{align*}
        \frac{\frac{\nu(\sigma)}{\mu(\sigma)} + 1}{\left| \frac{\nu(\sigma)}{\mu(\sigma)} - 1 \right|} \leq 1 + \frac{2}{\left| \frac{\nu(\sigma)}{\mu(\sigma)} - 1 \right|} \leq 1+\frac{2}{t}.
    \end{align*}
    Summing over all $\sigma \in \overline{M}$, we have
    \begin{align}\label{eq:M-bar}
        \sum\limits_{\sigma \in \overline{M}} \mu(\sigma) \left( \frac{\nu(\sigma)}{\mu(\sigma)} + 1 \right)^\alpha \leq \left( 1+\frac{2}{t} \right)^\alpha \sum\limits_{\sigma \in \overline{M}} \mu(\sigma) \left| \frac{\nu(\sigma)}{\mu(\sigma)} - 1 \right|^\alpha
        \leq 2 \left( 1+\frac{2}{t} \right)^\alpha \dvg{\nu}{\mu}{\chi^\alpha}.
    \end{align}

    Finally, adding the two inequalities~\eqref{eq:M} and~\eqref{eq:M-bar}, we have
    \begin{align}\label{eq:M-bar-bar}
        \sum\limits_{\sigma \in \Omega} \mu(\sigma) \left( \frac{\nu(\sigma)}{\mu(\sigma)} +1 \right) \leq \tp{\frac{2 \cdot \left( 2 + t \right)^\alpha}{b^{2\alpha}\theta^\alpha} + 2\left( 1+\frac{2}{t} \right)^\alpha} \dvg{\nu}{\mu}{\chi^\alpha} = g(t) \cdot \dvg{\nu}{\mu}{\chi^\alpha}.
    \end{align}
    The above inequality holds for any $0 \leq t \leq 1$. 
    We next find a value of $t$ to minimize the value of $g(t)$.
    Take the derivative of $g$, we get 
    \begin{align*}
        g'(t) = 2\alpha \left( 1+\frac{2}{t} \right)^{\alpha-1} \left( -\frac{2}{t^2} \right) + \frac{2\alpha}{b^{2\alpha }\theta^\alpha} (t+2)^{\alpha - 1} = 2\alpha (t+2)^{\alpha-1} \left( \frac{1}{b^{2\alpha}\theta^\alpha} - \frac{2}{t^{\alpha+1}}\right).
    \end{align*}    
    Thus when $t = \left( \frac{b^{2\alpha} \theta^\alpha}{2} \right)^{1/(\alpha+1)}$, where $ 0 < t < 1 $,the function $g$ obtains its minimum value:
    \begin{align*}
        g\left(t\right) &= \frac{2 \cdot \left( 2 + t \right)^\alpha}{b^{2\alpha}\theta^\alpha} + 2\left( 1+\frac{2}{t} \right)^\alpha
        = \frac{\left( 2 + t \right)^\alpha}{t^{\alpha+1}} + 2\left( 1+\frac{2}{t} \right)^\alpha\\
        &= \left( 1 + \frac{2}{t} \right)^\alpha \left( 2+ \frac{1}{t} \right) = \frac{(t+2)^\alpha(2t+1)}{t^{\alpha+1}}\\ 
        &= \frac{2}{b^{2\alpha} \theta^\alpha} \left( \left( \frac{b^{2\alpha} \theta^\alpha}{2} \right)^{1/(\alpha+1)} + 2 \right)^\alpha \left( 2 \left( \frac{b^{2\alpha} \theta^\alpha}{2} \right)^{1/(\alpha+1)} + 1 \right) = \frac{1}{B_{\alpha,b}(\theta)}.
    \end{align*}
    Combining the above equation with~\eqref{eq:M-bar-bar} implies that $ \sum_{\sigma \in \Omega} \mu(\sigma) \left( \frac{\nu(\sigma)}{\mu(\sigma)} +1 \right) \leq \frac{\dvg{\nu}{\mu}{\chi^\alpha}}{B_{\alpha,b}(\theta)}$. This proves the lower bound of the $\chi^\alpha$-divergence.
\end{proof}

\section{Algorithm for \texorpdfstring{$f$}{f}-divergence with small parameters distance}\label{sec:alg-abs}

In this section, we generalize the algorithm in \cite{feng2025approximating} to the case of small parameter distance. The new algorithm works for general $f$-divergence, where $f$ satisfies the following abstract condition.

\begin{condition}\label{cond:general-small-parameter-distance}
    Let $f$ be a function satisfying \Cref{ass:f-assumption}. There exists an function $F: \mathbb{R}^+ \to \mathbb{R}^+$ with $F(\zeta) \leq \mathrm{poly}(\zeta)$ such that for any $\zeta \geq 1$, any $x\in (-\frac{1}{2\zeta}, 0) \cup (0, \frac{1}{2\zeta})$, 
    \begin{align*}
        \frac{xf'(1 + \zeta x)}{f(1 + x)} \leq F\tp{\zeta}.
    \end{align*} 
\end{condition}

\begin{theorem}\label{lem:general-small-parameter-distance}
    Let $f$ be a function satisfying \Cref{cond:general-small-parameter-distance} with function $F$.
    There exists an algorithm such that given two Ising models $(G, J^\nu, h^\nu)$ and $(G, J^\mu, h^\mu)$, any $0 < \varepsilon < 1$, and $f,b,F$, if $\nu$ and $\mu$ are $b$-marginally bounded, $d_{\mathrm{par}}(\nu, \mu) < \theta = \frac{1}{10(n+3m)}$, and $\mu$ admits sampling oracle with cost functions $\oracle^{\mathrm{sp}}_G(\cdot)$, 
    then it returns a random number $\hat{D}$ in time $O(T \cdot \oracle^{\mathrm{sp}}_G(\frac{1}{100T}))$, where $T = O(\frac{F(8(n+3m)/b^2)^2(n+m)^2}{\varepsilon^2})$, $n = |V|$, $m = |E|$, such that 
    \begin{align*}
        \Pr{e^{-\varepsilon}\dvg{\nu}{\mu}{f} \leq \hat{D} \leq e^\varepsilon \dvg{\nu}{\mu}{f}} \geq \frac{2}{3}.
    \end{align*}
\end{theorem}

\begin{remark}
In \Cref{lem:general-small-parameter-distance}, the algorithm for small parameter case only requires sampling oracles for the input Ising model $(G, J^\mu, h^\mu)$ instead of sampling and approximate counting oracles for the whole family of Ising models $\+F$.
\end{remark}

By the definitions of $f$-divergence and Ising model, we have
    \begin{align*}
        \dvg{\nu}{\mu}{f} = \E[\mu]{f\left( \frac{\nu(\sigma)}{\mu(\sigma)} \right)} = \E[\mu]{f\left( \frac{w_\nu(\sigma)}{w_\mu(\sigma)} \cdot \frac{Z_\mu}{Z_\nu} \right)}.
    \end{align*}    
    Define the parameter 
    \[T \triangleq \left\lceil \frac{2^{12} \cdot 10^3 F(8(n+3m)/b^2)^2 (n+3m)^2}{b^4\varepsilon^2} \right\rceil.\]

    \begin{tcolorbox}[
        title = Algorithm for small parameter distance case, 
        colback=white,
        colframe=black,
        coltitle=black,
        colbacktitle=gray!20,
        fonttitle=\bfseries,
        boxrule=1pt,
        arc=3pt
    ] 
        \begin{itemize}[leftmargin=0.3cm]
            \item Call the sampling oracle of $\mu$ with TV-distance error $\frac{1}{100T}$ to obtain independent random samples $\hat\sigma_1,\hat\sigma_2,\ldots,\hat\sigma_T$. For each $i \in [T]$, compute $\hat{W}_i = \frac{w_\nu(\hat\sigma_i)}{w_\mu(\hat\sigma_i)}$.
            \item Return $\hat{D} = \frac{1}{T} \sum\limits_{i=1}^T f\left( {\hat{W}_i}/{\bar{W}} \right)$, where $\bar{W} = \frac{1}{T} \sum\limits_{i=1}^T \hat{W_i}$.
        \end{itemize}
    \end{tcolorbox}

    To analyze the algorithm, we define the random variable $W \triangleq \frac{w_\nu(X)}{w_\mu(X)}$, where $X \sim \mu$. By the definition of $W$, it is easy to verify that
    \begin{align}\label{eq:general-R-definition}
    \dvg{\nu}{\mu}{f} = \E[\mu]{f\left( \frac{w_\nu(\sigma)}{w_\mu(\sigma)} \cdot \frac{Z_\mu}{Z_\nu} \right)} = \E[]{f\left( \frac{W}{\E[]{W}}  \right)}, \quad \text{where } \E[]{W} = \frac{Z_\nu}{Z_\mu}
    \end{align}
    and we denote the $\mathbf{E}_{\mu}[f( \frac{W}{\E[]{W}})]$ as $D$ to simplify the notation.
    
    The following lemma controls the range and the variance of $W$.
    \begin{lemma}[\text{\cite{feng2025approximating}}]\label{lem:bound-for-W}
    Suppose $d = d_{\mathrm{par}}(\nu, \mu) < \frac{1}{10(n+3m)}$.
    For any $\sigma \in \Omega = \{\pm\}^V$,
    \begin{align*}
        1 -  2(n+3m) d \leq \exp\tp{- (n+3m) d }\leq \frac{w_\nu(\sigma)}{w_\mu(\sigma)} \leq \exp\tp{ (n+3m)d} \leq 1+(n+3m)d.
    \end{align*}
    As a consequence, $\Var[]{W} \leq (3n+9m)^2d^2$.
    \end{lemma}

    The above lemma can be verified easily by the definition of the parameter distance.
    For a formal proof, one can refer to \cite[Proof of Lemma 40]{feng2025approximating}\footnote{Proof of Lemma 37 in the arXiv version of the paper.}.

    To prove the correctness of the algorithm, we define a modified algorithm, which is only used in the analysis.
    Consider a modified algorithm that draw perfect independent samples $W_1,\ldots,W_T$ of $W$ instead of approximate samples $\hat{W}_1,\ldots,\hat{W}_T$. Let $\bar{W} = \frac{1}{T} \sum_{i=1}^T W_i$ and $\hat{D} = \frac{1}{T} \sum_{i=1}^T f\tp{{W_i}/{\bar{W}}}$ be the random variables computed in the same way as the algorithm but using the perfect samples. 
    Define the following random variable $\tilde{D}$ such that 
    \begin{align*}
        \tilde{D} = \frac{1}{T} \sum_{i=1}^T f\tp{\frac{W_i}{\E{W}}}.
    \end{align*}
    The algorithm does not know the exact value of $\E[]{W}$, so that it cannot compute $\tilde{D}$. However, we only use the random variable $\tilde{D}$ in the analysis. By~\eqref{eq:general-R-definition}, it is easy to verify that $\mathbf{E}[{\tilde{D}}] =D$. The following lemma shows the concentration of $\tilde{D}$.
    \begin{lemma}\label{lem:concentration-for-tilde-R}
    $\Pr[]{\vert \tilde{D} - D \vert \geq \frac{\varepsilon}{20} \cdot D } \leq 0.01$.
    \end{lemma}

    Next, since $\hat{D}$ and $\tilde{D}$ are both defined by the same set of random variables $W_1,W_2,\ldots,W_T$, they form a natural coupling. We show the following result in the coupling.
    \begin{lemma}\label{lem:coupling-for-hat-and-tilde-R}
    $\Pr[]{\vert \hat{D} - \tilde{D} \vert \geq \frac{\varepsilon}{20} \cdot D } \leq 0.01$.
    \end{lemma}

    \Cref{lem:concentration-for-tilde-R} and \Cref{lem:coupling-for-hat-and-tilde-R} will be proved later. Let us now assume that they hold. Note that $D = \dvg{\nu}{\mu}{f}$. By two lemmas and a simple union bound, it holds that 
    \begin{align}\label{eq:coupling-for-hat-and-tilde-R}
    \Pr[]{\vert \hat{D} - D \vert \geq \frac{\varepsilon}{10} \cdot \dvg{\nu}{\mu}{f} } \leq 0.02.
    \end{align}

    So far, the analysis is for perfect samples $W_1,W_2,\ldots,W_T$. We now consider the case that the algorithm draws approximate samples $\hat{W}_1,\ldots,\hat{W}_T$. Since our algorithm draws approximate samples from $\mu$ within total variation distance at most $\frac{1}{100 T}$, each $W_i$ and $\hat{W}_i$ can be coupled successfully with probability at least $1 - \frac{1}{100T}$. A simple union bound shows that the random variable $\hat{D}$ returned by algorithm satisfies the condition in the LHS of~\eqref{eq:coupling-for-hat-and-tilde-R} with probability at least $0.97$. This proves the correctness result in \Cref{lem:general-small-parameter-distance}. For the running time, it is easy to see that the algorithm runs in time $O(T \cdot (\oracle^{\mathrm{sp}}_G(\frac{1}{100T}) + n + m))$. By our assumption in \Cref{def:sample-approx-count-oracle}, $\oracle^{\mathrm{sp}}_G(\frac{1}{100T}) > n + m$.
    
    
    Now, we finish the proof of \Cref{lem:general-small-parameter-distance} by proving two technical lemmas.

    \begin{proof}[Proof of \Cref{lem:concentration-for-tilde-R}]
    Recall that $d = d_{\mathrm{par}}(\nu,\mu) < \frac{1}{10(n+3m)}$.
     By \Cref{lem:bound-for-W}, for any $i \in [T]$, it holds that $\e^{-2(n+3m)d} \leq \frac{W_i}{\E{W}} \leq \e^{2(n+3m)d}$ and
    \begin{align*}
        \Var{f\tp{\frac{W_i}{\E{W}}}} \leq \E{f\tp{\frac{W_i}{\E{W}}}^2} &\leq \max \left\{ f\tp{\e^{-2(n+3m)d}}^2, f\tp{\e^{2(n+3m)d}}^2 \right\}\\
        &\leq \max\{ f\tp{1-4(n+3m)d}^2, f\tp{1+4(n+3m)d}^2 \},
    \end{align*}
    where the last two inequalities hold due to the derivative bound in \Cref{ass:f-assumption}.

    Let us assume that $f\tp{1+4(n+3m)d} > f\tp{1-4(n+3m)d}$.  The other case can be solved with a symmetric argument.
    By Chebyshev's inequality and \Cref{cond:general-small-parameter-distance},
    \begin{align}\label{eq:EW}
        \Pr{\left| \tilde{D} - D \right| \geq \frac{\varepsilon}{20} \cdot f\tp{1+\frac{b^2}{2}d}} &\leq \frac{400\Var[]{\tilde{D}}}{\varepsilon^2 f\tp{1+\frac{b^2}{2}d}^2} \notag\\
    \text{$\left(\Var{\tilde{D}} \leq \frac{f\tp{1+4(n+3m)d}^2}{T}\right)$}\quad    &\leq \frac{400f\tp{1+4(n+3m)d}^2}{T\varepsilon^2 f\tp{1+\frac{b^2}{2}d}^2}\notag \\
    (\ast) \quad &\leq \frac{25600 (n+3m)^2 F(8(n+3m)/b^2)^2}{b^4T\varepsilon^2} \notag\\
    &\leq 0.01.
    \end{align}
    We explain how to obtain $(*)$. By mean value theorem, for any $x > 1$, $f(1+x) \leq f(1) + xf'(1 + x)=xf'(1+x)$, where the inequality holds because $f''(x) > 0$ for $x > 1$. Taking $x = 4(n+3m)d$ get the upper bound $\frac{6400(n+3m)^2 d^2 f'(1+4(n+3m)d)^2}{T\varepsilon^2 f(1+b^2d/2)^2}$. Inequalities $(\ast)$ holds by using \Cref{cond:general-small-parameter-distance} with $x = \frac{b^2}{2}d$, and $\zeta = 8(n+3m)/b^2$. Note that $4(n+3m)d < \frac{1}{2}$ and $\zeta \geq 1$.
    

    Thus, with probability at least $0.99$, we have
    \begin{align*}
        \left| \tilde{D} - D \right| \leq \frac{\varepsilon}{20} \cdot f\tp{1+\frac{b^2}{2}d} \overset{\text{\Cref{lem:f-lower-bound}}}{\leq} \frac{\varepsilon}{20} D. 
    \end{align*}

    The other case is $f\tp{1+4(n+3m)d} \leq f\tp{1-4(n+3m)d}$, where by a symmetric argument, we can show that $\Pr[]{|\tilde{D} - D| > \frac{\varepsilon}{20} \cdot f\tp{1-\frac{b^2}{2}d}} \leq 0.01$. The rest of the proof is the same.
    \end{proof}

    \begin{proof}[Proof of \Cref{lem:coupling-for-hat-and-tilde-R}]
    By comparing the definitions of $\hat{D}$ and $\tilde{D}$, to prove the lemma, we need to compare the terms $f\tp{W_i/\E[]{W}}$ and $f\tp{W_i/\bar{W}}$ in the summation. 
    By \Cref{lem:bound-for-W}, every sample $W_i$ is in $[\exp(-(n+3m)d), \exp((n+3m)d)]$. Then,
    \begin{align*}
        \left| \frac{W_i}{\E[]{W}} - \frac{W_i}{\bar{W}}  \right| = \frac{W_i}{\E[]{W} \bar{W}} \left| \bar{W} - \E{W} \right| \leq \e^{3(n+3m)d} \left| \bar{W} - \E{W} \right|.
    \end{align*}
    Note that $f(x)$ is continuous everywhere and differentiable except at $x = 1$.
    Suppose both $W_i/\bar{W}$ and $W_i/\E[]{W}$ are at least 1. Then, by the mean value theorem and the fact $f''(x) > 0$ for $x \neq 1$,
    \begin{align*}
        \left| f\tp{\frac{W_i}{\E[]{W}}} - f\tp{\frac{W_i}{\bar{W}}} \right| &\leq  \left| f'\tp{\e^{2(n+3m)d}} \right| \cdot  \left| \frac{W_i}{\E[]{W}} - \frac{W_i}{\bar{W}}  \right| \leq \left| f'\tp{\e^{2(n+3m)d}} \right|  \e^{3(n+3m)d} \left| \bar{W} - \E{W} \right|.
    \end{align*}
    If both $W_i/\bar{W}$ and $W_i/\E[]{W}$ are at most 1, then by a similar argument, we have
    \begin{align*}
        \left| f\tp{\frac{W_i}{\E[]{W}}} - f\tp{\frac{W_i}{\bar{W}}} \right| \leq \left| f'\tp{\e^{-2(n+3m)d}} \right|  \e^{3(n+3m)d} \left| \bar{W} - \E{W} \right|.
    \end{align*}
    The only remaining case is when $W_i/\bar{W}$ and $W_i/\E[]{W}$ are in different sides of 1. In this case, we can use $1$ as a bridge and apply the mean value theorem twice. Note that $f({W_i}/{\bar{W}}) - f(1) = f({W_i}/{\bar{W}}) \leq |f'({W_i}/{\bar{W}})|\cdot |{W_i}/{\bar{W}} - 1|$. It holds that
    \begin{align*}
        f\left(\frac{W_i}{\bar{W}}\right) \leq  \max\left\{ \left| f'\tp{\e^{-2(n+3m)d}}\right| , \left| f'\tp{\e^{2(n+3m)d}} \right| \right\} \cdot \left| \frac{W_i}{\bar{W}} - 1 \right|,
    \end{align*}
    By a similar argument, we can replace $W_i/\bar{W}$ with $W_i/\E[]{W}$ on both sides and get another inequality for $f(\frac{W_i}{\E[]{W}})$. Note that $|W_i/\bar  W  - W_i/\E[]{W}| = |W_i/\bar{W} - 1| + |W_i/\E[]{W} - 1|$ in this case. Combining the two inequalities, we have
    \begin{align*}
        \left| f\tp{\frac{W_i}{\E[]{W}}} - f\tp{\frac{W_i}{\bar{W}}} \right| &\leq f\tp{\frac{W_i}{\E[]{W}}} + f\tp{\frac{W_i}{\bar{W}}}\\
        & \leq  \max\left\{ \left| f'\tp{\e^{-2(n+3m)d}}\right| , \left| f'\tp{\e^{2(n+3m)d}} \right| \right\}\left| \frac{W_i}{\E[]{W}} - \frac{W_i}{\bar{W}}  \right| \\
        &\leq \max\left\{ \left| f'\tp{\e^{-2(n+3m)d}}\right| , \left| f'\tp{\e^{2(n+3m)d}} \right| \right\} \cdot \e^{3(n+3m)d} \left| \bar{W} - \E{W} \right|.
    \end{align*}
    Hence, combining the three cases, we have
    \begin{align*}
        \left| f\tp{\frac{W_i}{\E[]{W}}} - f\tp{\frac{W_i}{\bar{W}}} \right| &\leq \max\left\{ \left| f'\tp{\e^{-2(n+3m)d}}\right| , \left| f'\tp{\e^{2(n+3m)d}} \right| \right\} \cdot \e^{3(n+3m)d} \left| \bar{W} - \E{W} \right| \\
        &\leq 1.6 \max\left\{ \left| f'\tp{1-4(n+3m)d}\right| , \left| f'\tp{1+4(n+3m)d} \right| \right\} \left| \bar{W} - \E{W} \right|,
    \end{align*}
    where the last inequality holds $f''(x) > 0$ for all $x \neq 1$ and $\e^{3(n+3m)d} \leq 1.6$ by the assumption that $d < \frac{1}{10(n+3m)}$.
    Since the above inequality holds for all $i \in [T]$, we have
    \begin{align*}
        \left\vert \hat{D} - \tilde{D} \right\vert &=   \left| \frac{1}{T} \sum_{i=1}^T f\tp{\frac{W_i}{\bar{W}}} - \frac{1}{T} \sum_{i=1}^T f\tp{\frac{W_i}{\E[]{W}}} \right| \leq \frac{1}{T} \sum_{i=1}^T \left| f\tp{\frac{W_i}{\bar{W}}} - f\tp{\frac{W_i}{\E[]{W}}} \right|\\ 
        &\leq 3.2 \max\left\{ - f'\tp{1-4(n+3m)d} ,  f'\tp{1+4(n+3m)d} \right\} \left| \bar{W} - \E{W} \right|.
    \end{align*}

    Assume the case $f'\tp{1+4(n+3m)d} \geq - f'\tp{1-4(n+3m)d}$.
    We bound the difference between $\bar{W}$ and $\E{W}$. 
    Since $\bar{W}$ is the average of $T$ independent samples of $W$, by \Cref{lem:bound-for-W}, $\textbf{Var}[\bar{W}] \leq \frac{(3n+9m)^2}{T} \cdot d^2$. Recall that $d = d_{\mathrm{par}}(\nu, \mu)$. By using Chebyshev's inequality on $\bar{W}$, we have
    \begin{align*}
        \Pr{\left| \bar{W} - \E{W} \right| \geq \frac{\varepsilon f\tp{1+\frac{b^2}{2}d}}{20 \cdot 1.6 f'(1+4(n+3m)d)} } &\leq \frac{\Var[]{\bar{W}}\cdot 1024 f'(1+4(n+3m)d)^2}{\varepsilon^2 f\tp{1+\frac{b^2}{2}d}^2}\notag\\
        &\leq \frac{(3n+9m)^2d^2\cdot 1024 f'(1+4(n+3m)d)^2}{T\varepsilon^2 f\tp{1+\frac{b^2}{2}d}^2}\\
     \text{$\left(\text{\Cref{cond:general-small-parameter-distance}}: x = \frac{b^2}{2}d, \zeta = 8(n+3m)/b^2\right)$} \quad  &\leq \frac{(3n+9m)^2\cdot 4096 F(8(n+3m)/b^2)^2}{b^4 T\varepsilon^2}\leq \frac{1}{100},    
    \end{align*}
    where the last inequality comes from the definition of $T$.
    Hence,  with probability at least $0.99$, 
    \begin{align*}
        \left\vert \hat{D} - \tilde{D} \right\vert \leq \frac{\varepsilon}{20} \cdot f\tp{1+\frac{b^2}{2}d} \overset{\text{\Cref{lem:f-lower-bound}}}{\leq} \frac{\varepsilon}{20} D.
    \end{align*}

    The other case is $f'\tp{1+4(n+3m)d} < - f'\tp{1-4(n+3m)d}$. By a symmetric argument, we can show that with probability at most $0.01$, $\left| \bar{W} - \E{W} \right| \leq \frac{\varepsilon f(1-\frac{b^2}{2}d)}{20 \cdot 1.6 |f'(1-4(n+3m)d)|}$. The rest of the proof is the same.
    \end{proof}

\section{Algorithm for \texorpdfstring{$\chi^\alpha$}{chi-alpha}-divergence}\label{sec:alg-chi-alpha}

Our algorithm first computes the marginal lower bound $b$ in time $O(n+m)$. Due to the conditional independence, the algorithm only need to enumerate all $v \in V$, any $c \in \{-1,+1\}$, and find a worst pinning $\tau = \tau(v,c)$ on all neighbor of $v$. Specifically, for any neighbor $u$ of $v$, if $J_{uv} > 0$, then fix $\tau_u = -c$; otherwise, fix $\tau_u = c$. The algorithm compute $\mu_v^\tau(c)$. Let $b$ be the minimum value over all $v \in V$ and $c \in \{-1,+1\}$. We now assume that the value $b$ is known by the algorithm.

Let $\theta = \frac{1}{10(n + 3m)}$ be the threshold parameter. We give two algorithms in \Cref{sec:algorithm-odd-alpha} and \Cref{sec:algorithm-small-parameter-distance} depending on whether $d_{\mathrm{par}}(\nu,\mu) > \theta$ or not. We prove \Cref{thm:approx-chi-alpha-divergence} in \Cref{sec:putting-everything-together}.

\subsection{Large parameters distance case}\label{sec:algorithm-odd-alpha}
\begin{lemma}
    Let $\alpha \geq 1$ be an integer and $b \in (0,1)$ be a constant.
    There exists an algorithm such that given two Ising models $(G, J^\nu, h^\nu)$ and $(G, J^\mu, h^\mu)$, any $0 < \varepsilon < 1$, and $\alpha,b$, if $\nu$ and $\mu$ are $b$-marginally bounded, $d_{\mathrm{par}}(\nu, \mu) > \theta$, and every Ising model in $\+F$ admits sampling and approximate counting oracles with cost function $\oracle^{\mathrm{sp}}_G(\cdot)$ and $\oracle^{\mathrm{ct}}_G(\cdot)$ respectively, then it returns a random $\hat{D}$ in time $O_{\alpha,b}\left( \oracle^{\mathrm{ct}}_G(\delta) + T \cdot \left( \oracle^{\mathrm{sp}}_G \left(\frac{1}{200T(\alpha + 1)}\right) \right) \right)$, where $\delta = \Theta_{\alpha,b}(\theta^\alpha\varepsilon)$ and $T = \Theta_{\alpha,b}(\frac{1}{\varepsilon^2 \theta^\alpha})$ such that 
    \begin{align*}
        \Pr{\e^{-\varepsilon}\dvg{\nu}{\mu}{\chi^\alpha} \leq \hat{D} \leq \e^\varepsilon \dvg{\nu}{\mu}{\chi^\alpha}} \geq \frac{2}{3}.
    \end{align*}
\end{lemma}

By the definition of $\chi^\alpha$-divergence, we can rewrite 
\begin{align}\label{eq:expansion-for-odd-alpha}
    \dvg{\nu}{\mu}{\chi^\alpha} &= \frac{1}{2}\sum_{\sigma: \nu(\sigma) > \mu(\sigma)}\mu(\sigma) \tp{\frac{\nu(\sigma)}{\mu(\sigma)} - 1}^\alpha + \frac{1}{2}\sum_{\sigma: \nu(\sigma) < \mu(\sigma)} \mu(\sigma)\tp{1 - \frac{\nu(\sigma)}{\mu(\sigma)}}^\alpha\\
    &=\frac{1}{2} \sum_{k=0}^{\alpha} (-1)^{\alpha - k} \binom{\alpha}{k} \left( \sum\limits_{\sigma: \nu(\sigma) > \mu(\sigma)} \frac{\nu^k(\sigma)}{\mu^{k-1}(\sigma)} + (-1)^\alpha \sum\limits_{\sigma: \nu(\sigma) < \mu(\sigma)} \frac{\nu^k(\sigma)}{\mu^{k-1}(\sigma)} \right).
\end{align}

Recall that $Z_k$ is the partition function of the Ising model $(G, J^{(k)}, h^{(k)})$, where $J^{(k)} \triangleq kJ^\nu - (k-1)J^\mu$ and $h^{(k)} \triangleq kh^\nu - (k-1)h^\mu$. Note that $Z_\mu = Z_0$ and $Z_\nu = Z_1$. Define two random variables $W^+_k$ and $W^-_k$ as follows. Let $\pi^{(k)}$ be the Gibbs distributions of the Ising model $(G, J^{(k)}, h^{(k)})$.
    \begin{align*}
        W^+_k = \mathbf{1}[\nu(X) > \mu(X)] \frac{Z_0^{k-1} \cdot Z_k}{Z_1^k}, \quad\text{where}\quad X \sim \pi^{(k)}.\\
        W^-_k = \mathbf{1}[\nu(Y) < \mu(Y)] \frac{Z_0^{k-1} \cdot Z_k}{Z_1^k}, \quad\text{where}\quad Y \sim \pi^{(k)}.
    \end{align*} 
 The expectation of $W^+_k$ can be calculated as follows
 \begin{align*}
    \E{W^+_k} = \sum\limits_{\sigma: \nu(\sigma) > \mu(\sigma)} \pi^{(k)}(\sigma) \frac{Z_0^{k-1} \cdot Z_k}{Z_1^k} &= \sum\limits_{\sigma: \nu(\sigma) > \mu(\sigma)} \frac{w_\nu^k(\sigma)/w_\mu^{k-1}(\sigma)}{Z_k} \frac{Z_0^{k-1} \cdot Z_k}{Z_1^k}=\sum\limits_{\sigma: \nu(\sigma) > \mu(\sigma)} \frac{\nu^k(\sigma)}{\mu^{k-1}(\sigma)}.
 \end{align*}
 Similarly, we have $\E{W^-_k} = \sum_{\sigma: \nu(\sigma) < \mu(\sigma)} \frac{\nu^k(\sigma)}{\mu^{k-1}(\sigma)}$. In a high level, our algorithm draws approximate samples of $W^+_k$ and $W^-_k$ and estimate $\mathbf{E}[W^+_k] + (-1)^\alpha \mathbf{E}[W^-_k]$. Define the parameter 
    \begin{align*}
        T &\triangleq \left\lceil \frac{8 \cdot 10^4 (\alpha+1)}{\varepsilon^2 B_{\alpha,b}(\theta)^2} \right\rceil \quad \text{and} \quad \delta \triangleq \frac{B_{\alpha,b}(\theta) \varepsilon}{20 (\alpha+1)},
    \end{align*}
    where the function $B_{\alpha,b}(\theta)$ is defined in \Cref{lem:error-bound}.
    
    \begin{tcolorbox}[
        title = Algorithm for large parameter distance case, 
        colback=white,
        colframe=black,
        coltitle=black,
        colbacktitle=gray!20,
        fonttitle=\bfseries,
        boxrule=1pt,
        arc=3pt
        ]
        \begin{itemize}[leftmargin=0.3cm]
            \item For $0\leq k \leq \alpha$, call the approximate counting oracle on Ising model $(G, J^{(k)}, h^{(k)})$ for $O(\log \alpha)$ times independently with relative error $\delta$, and take the median as $\hat{Z}_k$.
            \item For each $k$ from $0$ to $\alpha$ do:
            \begin{enumerate}
                \item Draw $2T$ samples $\hat{\sigma}_1,\ldots,\hat{\sigma}_{2T}\sim \pi^{(k)}$ by the sampling oracle with error $\frac{1}{200T(\alpha + 1)}$. 
                \item For each $i \in [T]$, compute two values 
                $$\hat{W}^+_{k,i} = \mathbf{1}[\hat{\nu}(\hat{\sigma}_i) > \hat{\mu}(\hat{\sigma}_i)] \frac{\hat{Z}_0^{k-1} \cdot \hat{Z}_k}{\hat{Z}_1^k} \text{ and } \hat{W}^-_{k,i} = \mathbf{1}[\hat{\nu}(\hat{\sigma}_{T+i}) < \hat{\mu}(\hat{\sigma}_{T+i})] \frac{\hat{Z}_0^{k-1} \cdot \hat{Z}_k}{\hat{Z}_1^k},$$
                where for any $x \in \Omega$, $\hat{\nu}(x) = \frac{w_\nu(x)}{\hat{Z}_1}$ and $\hat{\mu}(x) = \frac{w_\mu(x)}{\hat{Z}_0}$. 
            \end{enumerate}
            \item Return $\hat{D} = \frac{1}{2} \sum_{k=0}^{\alpha} (-1)^{\alpha - k} \binom{\alpha}{k} \left( \frac{1}{T} \sum_{i=1}^T \hat{W}^+_{k,i} + (-1)^\alpha \frac{1}{T} \sum_{i=1}^T \hat{W}^-_{k,i} \right)$.
        \end{itemize}
    \end{tcolorbox}
    The total running time of the algorithm is $O(\alpha \log \alpha) \oracle^{\mathrm{ct}}_G(\delta) + O(T)\cdot (\oracle^{\mathrm{sp}}_G \left(\frac{1}{200 T} \right) + n + m)$. Note that $\delta = \Theta_{\alpha,b}(\theta^\alpha\varepsilon)$ and $T = \Theta_{\alpha,b}(\frac{1}{\varepsilon^2 \theta^\alpha})$. We may assume that sampling and counting oracle needs to read the whole graph $G$. Hence, $O(n+m)$ is dominated by the time of sampling. This proves the running time of the algorithm.

    Without loss of generality, we can make following assumptions.
    \begin{itemize}
        \item All of $\hat{Z}_k$ satisfies $\e^{-\delta} Z_k \leq \hat{Z}_k \leq \e^\delta Z_k$ (which happens with probability at least $0.99$).
        \item We can replace approximate samples $\hat{\sigma}_1,\ldots,\hat{\sigma}_{2T}$ with perfect samples $\sigma_1,\ldots,\sigma_{2T}$ for any $\pi^{(k)}$ with $0 \leq k \leq \alpha$. Since the sampling TV-distance error is $\frac{1}{200T(\alpha + 1)}$, by a simple coupling argument, all samples $\hat{\sigma}_i$ can be coupled perfectly with $\sigma_i$ with probability at least $0.99$.
    \end{itemize}
    Using a simple union bound, the above two assumption only contribute at most 0.02 failure probability of the algorithm. We always assume that the assumption holds in the following analysis.

    To prove the correctness of the algorithm, we need to analyze two types of errors. The first type of error is from $\hat{\nu}$, $\hat{\mu}$, $\hat{Z}_k$ for $k \in [0, \alpha]$. The second type of error is from using average of $\hat{W}^+_{k,i}$ and $\hat{W}^-_{k,i}$ to estimate the expectation $\mathbf{E}[W^+_k] + (-1)^\alpha \mathbf{E}[W^-_k]$.

    To bound the first type of error, the expectation of the output $\hat{D}$ can be written as
    \begin{align*}
        \E[]{\hat D} = \frac{1}{2} \sum_{k=0}^{\alpha} (-1)^{\alpha - k} \binom{\alpha}{k} \left( \sum\limits_{\sigma: \hat{\nu}(\sigma) > \hat{\mu}(\sigma)} \pi^{(k)}(\sigma) \frac{\hat Z_0^{k-1} \cdot \hat Z_k}{\hat Z_1^k} + (-1)^\alpha \sum\limits_{\sigma: \hat{\nu}(\sigma) < \hat{\mu}(\sigma)} \pi^{(k)}(\sigma) \frac{\hat Z_0^{k-1} \cdot \hat Z_k}{\hat Z_1^k} \right).
    \end{align*}
    We prove two lemmas. The first lemma bounds the error between $\mathbf{E}[\hat D]$ and $\dvg{\nu}{\mu}{\chi^\alpha}$ and the second lemma bounds the concentration error between $\hat D$ and $\mathbf{E}[\hat D]$.

    \begin{lemma}[Error of expectation]\label{lem:error-from-hat-nu-hat-mu}
    \begin{align*}
        \left \vert \E[]{\hat D} - \dvg{\nu}{\mu}{\chi^\alpha} \right \vert \leq \frac{\epsilon}{10} \dvg{\nu}{\mu}{\chi^\alpha}.
    \end{align*}
    \end{lemma}

    \begin{proof}
    To bound the difference, we define a middle term 
    \begin{align*}
     &\tdvg{\nu}{\mu}{\chi^\alpha} = \frac{1}{2} \sum_{k=0}^{\alpha} (-1)^{\alpha - k} \binom{\alpha}{k} \left( \sum\limits_{\sigma: \hat\nu(\sigma) > \hat\mu(\sigma)} \frac{\nu^k(\sigma)}{\mu^{k-1}(\sigma)} + (-1)^\alpha \sum\limits_{\sigma: \hat\nu(\sigma) < \hat\mu(\sigma)} \frac{\nu^k(\sigma)}{\mu^{k-1}(\sigma)} \right)\\
     =\,& \frac{1}{2}\sum_{\sigma: \hat\nu(\sigma) > \hat\mu(\sigma)}\mu(\sigma) \tp{\frac{\nu(\sigma)}{\mu(\sigma)} - 1}^\alpha + \frac{1}{2}\sum_{\sigma: \hat\nu(\sigma) < \hat\mu(\sigma)}\mu(\sigma) \tp{1 - \frac{\nu(\sigma)}{\mu(\sigma)}}^\alpha.
    \end{align*}
    Compared to the true divergence in~\eqref{eq:expansion-for-odd-alpha}, $\tdvg{\nu}{\mu}{\chi^\alpha}$ only replaces $\nu$ and $\mu$ with $\hat\nu$ and $\hat\mu$ in the index of summing notation. We next bound $\vert \dvg{\nu}{\mu}{\chi^\alpha} - \tdvg{\nu}{\mu}{\chi^\alpha} \vert$ and $\vert \tdvg{\nu}{\mu}{\chi^\alpha} - \hdvg{\nu}{\mu}{\chi^\alpha} \vert$.

    The error between $\dvg{\nu}{\mu}{\chi^\alpha}$ and $\tdvg{\nu}{\mu}{\chi^\alpha}$ comes from that $\tdvg{\nu}{\mu}{\chi^\alpha}$ may miss classify some $\sigma \in \Omega$ into the wrong set. By the definition of $\tdvg{\nu}{\mu}{\chi^\alpha}$, we have the following fact.
    If $\nu(\sigma) > \mu(\sigma)$ (resp. $\nu(\sigma) < \mu(\sigma)$), by~\eqref{eq:expansion-for-odd-alpha}, the contribution of $\sigma$ to the divergence $\dvg{\nu}{\mu}{\chi^\alpha}$ is $\mu(\sigma)(\frac{\mu(\sigma)}{\nu(\sigma)}-1)^\alpha$ (resp. $\mu(\sigma)(1-\frac{\mu(\sigma)}{\nu(\sigma)})^\alpha$). However, in $ \tdvg{\nu}{\mu}{\chi^\alpha}$, it may be possible that $\hat\nu(\sigma) < \hat\mu(\sigma)$ (resp. $\hat\nu(\sigma) > \hat\mu(\sigma)$), then the contribution of $\sigma$ to the $\tdvg{\nu}{\mu}{\chi^\alpha}$ is $\mu(\sigma)(1-\frac{\mu(\sigma)}{\nu(\sigma)})^\alpha$ (resp. $\mu(\sigma)(\frac{\mu(\sigma)}{\nu(\sigma)}-1)^\alpha$). 
    In both cases, the additive error is at most 
    \begin{align*}
    \mu(\sigma)\left\vert \tp{\frac{\nu(\sigma)}{\mu(\sigma)} - 1}^\alpha - \tp{1 - \frac{\nu(\sigma)}{\mu(\sigma)}}^\alpha \right\vert \leq 2\mu(\sigma) \left\vert \frac{\nu(\sigma)}{\mu(\sigma)} - 1 \right\vert^\alpha,
    \end{align*}
    where the error occurs only if $\alpha$ is odd.

    For any $\sigma$, $\nu(\sigma) = \frac{w_{\nu}(\sigma)}{Z_\nu}$ and $\mu(\sigma) = \frac{w_{\mu}(\sigma)}{Z_\mu}$. By assumption, our estimation achieves $\e^{-\delta} \nu(\sigma) \leq \hat{\nu}(\sigma) \leq \e^\delta \nu(\sigma)$, and so does $\hat{\mu}(\sigma)$. Hence, the miss classification only happens when $\frac{\nu(\sigma)}{\mu(\sigma)} \in [\e^{-2\delta}, \e^{2\delta}]$. Let $B$ denote the set of all such $\sigma$'s. Note that $|\frac{\nu(\sigma)}{\mu(\sigma)} - 1| \leq \max\{1-\e^{-2\delta},\e^{2\delta}- 1\} \leq 4 \delta$. We have
    \begin{align*}
        \vert \dvg{\nu}{\mu}{\chi^\alpha} - \tdvg{\nu}{\mu}{\chi^\alpha} \vert &\leq 2 \cdot \frac{1}{2} \sum_{\sigma \in B} \mu(\sigma) \left\vert \frac{\nu(\sigma)}{\mu(\sigma)} - 1 \right\vert^\alpha \leq (4 \delta)^\alpha \leq 4 \delta,
    \end{align*}
    where the last inequality follows from the fact that $4\delta < 1$.
    
    Next, we bound the error between $\tdvg{\nu}{\mu}{\chi^\alpha}$ and $\hdvg{\nu}{\mu}{\chi^\alpha}$. 
    By our assumption, we have $\e^{-\delta} Z_k \leq \hat{Z}_k \leq \e^\delta Z_k$ for all $k$. Note that $\pi^{(k)}(\sigma) = \frac{w_\nu^k(\sigma)/w_\mu^{k-1}(\sigma)}{Z_k}$, $Z_0 = Z_\mu$ and $Z_1 = Z_\nu$. The following observation is easy to verify. For any $\sigma$, it holds that
    \begin{align}\label{eq:error-from-hat-nu-hat-mu}
        &\e^{-2\alpha\delta} \frac{\nu^k(\sigma)}{\mu^{k-1}(\sigma)}  \leq  \pi^{(k)}(\sigma) \frac{\hat Z_0^{k-1} \cdot \hat Z_k}{\hat Z_1^k} \leq \e^{2\alpha\delta} \frac{\nu^k(\sigma)}{\mu^{k-1}(\sigma)} \notag\\
   \text{by$\left(\delta < \frac{1}{2\alpha}\right)$}\quad     \implies\,& \left \vert \pi^{(k)}(\sigma) \frac{\hat Z_0^{k-1} \cdot \hat Z_k}{\hat Z_1^k} - \frac{\nu^k(\sigma)}{\mu^{k-1}(\sigma)} \right \vert \leq 4\alpha\delta \frac{\nu^k(\sigma)}{\mu^{k-1}(\sigma)}.
    \end{align}
    By the triangle inequality, the difference $\vert \tdvg{\nu}{\mu}{\chi^\alpha} - \mathbf{E}[\hat D] \vert$ is at most
    \begin{align*}
        & \frac{1}{2}\sum_{k=0}^{\alpha} \binom{\alpha}{k} \left( \sum\limits_{\sigma: \hat\nu(\sigma) > \hat\mu(\sigma)} \left\vert \pi^{(k)}(\sigma) \frac{\hat Z_0^{k-1} \cdot \hat Z_k}{\hat Z_1^k} - \frac{\nu^k(\sigma)}{\mu^{k-1}(\sigma)}  \right\vert + \sum\limits_{\sigma: \hat\nu(\sigma) < \hat\mu(\sigma)} \left\vert \pi^{(k)}(\sigma) \frac{\hat Z_0^{k-1} \cdot \hat Z_k}{\hat Z_1^k} - \frac{\nu^k(\sigma)}{\mu^{k-1}(\sigma)}  \right\vert \right)\\
         &\leq 2 \alpha \delta \sum_{k=0}^\alpha \binom{\alpha}{k} \sum_{\sigma \in \Omega} \frac{\nu^k(\sigma)}{\mu^{k-1}(\sigma)} = 2 \alpha \delta \sum_{\sigma \in \Omega} \mu(\sigma) \left( \frac{\nu(\sigma)}{\mu(\sigma)} + 1 \right)^\alpha \leq \frac{2\alpha \delta}{B_{\alpha,b}(\theta)} \dvg{\nu}{\mu}{\chi^\alpha}.
    \end{align*}
    where the first inequality follows from~\eqref{eq:error-from-hat-nu-hat-mu} and the last inequality follows from \Cref{lem:error-bound}.

    Using the middle term $\tdvg{\nu}{\mu}{\chi^\alpha}$ as a bridge, we have
    \begin{align*}
        \left \vert \E[]{\hat D} - \dvg{\nu}{\mu}{\chi^\alpha}  \right \vert \leq \left(4 \delta + \frac{2 \alpha \delta}{B_{\alpha,b}(\theta)} \right) \dvg{\nu}{\mu}{\chi^\alpha}.
    \end{align*}
    We set $\delta = \frac{B_{\alpha,b}(\theta) \varepsilon}{20 (\alpha+1)}$ and remark that $B_{\alpha,b}(\theta) \leq \frac{1}{2}$ by definition in \Cref{lem:error-bound}, $\tp{4\delta  + \frac{2 \alpha \delta}{B_{\alpha,b}(\theta)}} \leq \frac{\varepsilon}{10}$, and thus we finish the proof.
    \end{proof}

    \begin{lemma}[concentration error]\label{lem:computational-error}
    \begin{align*}
       \Pr[]{\left \vert \hat D - \E[]{\hat D} \right \vert \leq \frac{\varepsilon}{10} \dvg{\nu}{\mu}{\chi^\alpha}} \geq 0.99.
    \end{align*}
    \end{lemma}
    \begin{proof}
    In this proof, we fix the values of all $\hat{Z}_k$ returned by the approximate counting oracle at the first step.
    Define the random variable $\hat{W}^+_k \triangleq \frac{1}{T} \sum_{i=1}^T \hat{W}^+_{k,i}$, then $\mathbf{E}[\hat{W}^+_k] = \mathbf{E}[\hat{W}^+_k]$.
    By definition of $\hat{W}_{k,i}^+$, the random variable $\hat{W}^+_k $ is within the range $[0, (\hat{Z}_0^{k-1} \cdot \hat{Z}_k)/\hat{Z}_1^k]$,
    \begin{align*}
        \Var{\hat{W}^+_k} = \frac{1}{T} \Var{\hat{W}^+_{k,i}} \leq \frac{1}{T} \cdot \left( \frac{\hat{Z}_0^{k-1} \cdot \hat{Z}_k}{\hat{Z}_1^k} \right)^2.
    \end{align*}
    By Chebyshev's inequality on $\bar{W}^+_k$, we have
    \begin{align*}
        \Pr{\left| \bar{W}^+_k - \E{\hat{W}^+_k} \right| \geq \frac{\varepsilon B_{\alpha,b}(\theta)}{20} \cdot \frac{\hat{Z}_0^{k-1} \cdot \hat{Z}_k}{\hat{Z}_1^k}} &\leq \frac{400 \cdot \Var{\bar{W}^+_k}}{\varepsilon^2 B_{\alpha,b}(\theta)^2 \cdot \left( {\hat{Z}_0^{k-1} \cdot \hat{Z}_k}/{\hat{Z}_1^k} \right)^2} \leq \frac{1}{200(\alpha + 1)},
    \end{align*}
    where the last inequality due to the definition of $T$.
    Similarly, we have 
    \begin{align*}
        \Pr{\left| \bar{W}^-_k - \E{\hat{W}^-_k} \right| \geq \frac{\varepsilon B_{\alpha,b}(\theta)}{20} \cdot \frac{\hat{Z}_0^{k-1} \cdot \hat{Z}_k}{\hat{Z}_1^k} } \leq \frac{1}{200(\alpha + 1)}.
    \end{align*}
    By a union bound, with probability at least $0.99$, it holds for all $0 \leq k \leq \alpha$,
    \begin{align*}
        \left| \bar{W}^+_k - \E{\hat{W}^+_k} \right| + \left| \bar{W}^-_k - \E{\hat{W}^-_k} \right| \leq \frac{\varepsilon B_{\alpha,b}(\theta)}{10} \cdot \frac{\hat{Z}_0^{k-1} \cdot \hat{Z}_k}{\hat{Z}_1^k}.
    \end{align*}

    Finally, by the triangle inequality, we have that with probability at least $0.99$,
    \begin{align*}
        \left| \hat{D} - \E{\hat{D}} \right| &\leq \frac{1}{2} \sum_{k=0}^\alpha \binom{\alpha}{k} \tp{\left| \bar{W}^+_k - \E{\hat{W}^+_k} \right| + \left| \bar{W}^-_k - \E{\hat{W}^-_k} \right|}\\
        &\leq \frac{\varepsilon B_{\alpha,b}(\theta)}{20} \sum_{k=0}^\alpha \binom{\alpha}{k} \frac{\hat{Z}_0^{k-1} \cdot \hat{Z}_k}{\hat{Z}_1^k}\leq \frac{\varepsilon B_{\alpha,b}(\theta)}{20} \cdot e^{2\alpha \delta} \sum_{k=0}^\alpha \binom{\alpha}{k} \frac{Z_0^{k-1} \cdot Z_k}{Z_1^k}\\ 
        &= \frac{\varepsilon \cdot e^{2\alpha \delta}}{20} \cdot B_{\alpha,b}(\theta) \sum_{\sigma \in \{\pm\}^V} \mu(\sigma) \tp{\frac{\nu(\sigma)}{\mu(\sigma)} + 1}^\alpha\\
   \text{(by \Cref{lem:error-bound})}\quad     &\leq \frac{\varepsilon \cdot e^{2\alpha \delta}}{20} \dvg{\nu}{\mu}{\chi^\alpha} \leq \frac{\varepsilon}{10} \dvg{\nu}{\mu}{\chi^\alpha}. \qedhere
    \end{align*}
    \end{proof}
The correctness of the algorithm follows from \Cref{lem:error-from-hat-nu-hat-mu} and \Cref{lem:computational-error}.

\subsection{Small parameters distance case}\label{sec:algorithm-small-parameter-distance}
With \Cref{lem:general-small-parameter-distance}, we only need to verify \Cref{cond:general-small-parameter-distance} for $f = \frac{1}{2} |x-1|^\alpha$. By definition,  
\begin{align*}
    \frac{xf'(1+\zeta x)}{f(1+x)} = \frac{\frac{\alpha}{2}|x|^\alpha|\zeta|^{\alpha-1}}{\frac{1}{2}|x|^\alpha} = \alpha \zeta^{\alpha-1}.
\end{align*}
Set $F(\zeta) = \alpha \zeta^{\alpha-1}$, $F(\zeta)$ satisfies \Cref{cond:general-small-parameter-distance}. Thus 
\begin{align*}
    F(8(n+3m)/b^2) =\frac{2^{3\alpha-3}\alpha(n+3m)^{\alpha-1}}{b^{2\alpha-2}}.
\end{align*}

Hence, using \Cref{lem:general-small-parameter-distance}, there is an algorithm with running time 
$O(T\cdot \oracle^{\mathrm{sp}}_G (\frac{1}{100T}))$, where $T = O(\frac{2^{6\alpha - 6} \alpha^2 (n+3m)^{2\alpha}}{b^{4\alpha - 4} \varepsilon^2}) = O_{\alpha,b}(\frac{(n+3m)^{2\alpha}}{\varepsilon^2})$, 
for small parameter distance case $d_{\mathrm{par}}(\nu,\mu) < \theta = \frac{1}{10(n+3m)}$.

    \subsection{Putting everything together (Proof of \texorpdfstring{\Cref{thm:approx-chi-alpha-divergence}}{})}\label{sec:putting-everything-together}
    
    \Cref{thm:approx-chi-alpha-divergence} follows from \Cref{lem:error-from-hat-nu-hat-mu} and \Cref{lem:computational-error}. 
    Define parameter $T = \Theta_{\alpha,b}(\frac{(n+m)^{2\alpha}}{\varepsilon^2})$ and $\delta = \Theta_{\alpha,b}(\frac{\varepsilon}{(n+m)^\alpha})$. By choosing constants, we can make $T$ large enough and $\delta$ small enough.
    The preprocessing time for computing $b$ and $d_{\mathrm{par}}(\nu,\mu)$ is $O(n+m)$, which is dominated by the time for sampling and approximate counting.
    The running time of the whole algorithm is at most
    \begin{align*}
        &\max \left\{  O_{\alpha,b}\left( \oracle^{\mathrm{ct}}_G(\delta) + T \cdot \oracle^{\mathrm{sp}}_G \left(\frac{1}{200T(\alpha + 1)}\right) \right), O_{\alpha,b}\tp{T \cdot \oracle^{\mathrm{sp}}_G\left(\frac{1}{100T}\right)}\right\}\\
        \leq\,&  O_{\alpha,b}\tp{\oracle^{\mathrm{ct}}_G\tp{\frac{\eta_{\alpha,b} \cdot \varepsilon}{(n+m)^\alpha}} +\frac{(n+m)^{2\alpha}}{\varepsilon^2} \cdot  \oracle^{\mathrm{sp}}_G\left(\frac{\eta_{\alpha,b} \cdot \varepsilon^2}{(n+m)^{2\alpha}}\right)},
    \end{align*}
    where the last inequality comes from the fact that both $\oracle^{\mathrm{ct}}_G$ and $\oracle^{\mathrm{sp}}_G$ are non-increasing functions and $\eta_{\alpha,b} > 0$ is a small enough constant depending only on $\alpha$ and $b$.

\section{Algorithms for other \texorpdfstring{$f$}{f}-divergences}\label{sec:other-f-divergences}

In this section, we briefly show the algorithms for other $f$-divergences. 
The algorithms prove the results in \Cref{thm:kl-divergence}, \Cref{thm:alpha-divergence}, and \Cref{thm:squared-hellinger-distance}. 

Again, the algorithm first computes the parameter distance $d_{\mathrm{par}}(\nu,\mu)$. Then, it compares $d_{\mathrm{par}}(\nu,\mu)$ to the threshold $\theta = \frac{1}{10(n+3m)}$. If $d_{\mathrm{par}}(\nu,\mu) \leq \theta$, algorithms are given in \Cref{sec:alg-small-parameter-distance}. Otherwise $d_{\mathrm{par}}(\nu,\mu) > \theta$, algorithms are given in \Cref{sec:alg-large-parameter-distance}.

\subsection{Algorithms for small parameter distance}\label{sec:alg-small-parameter-distance}


In this case, we use the abstract algorithm in \Cref{lem:general-small-parameter-distance}. We only need to verify \Cref{cond:general-small-parameter-distance} for $f$-divergences.
The following lemma gives a sufficient condition for \Cref{cond:general-small-parameter-distance}. 

\begin{lemma}\label{lem:sufficient-condition-for-small-parameter-distance}
    Suppose $f: \mathbb{R}_{>0} \to \mathbb{R}_{\geq 0}$ is twice differentiable at every $x > 0$ such that $f''(x) > 0$ and $f'(1) = f(1) = 0$. 
    If there exists two positive constants $L,U > 0$ such that $L \leq f''(x) \leq U$ for any $x \in [\frac{1}{2},\frac{3}{2}]$. Then, the function $f$ satisfies \Cref{cond:general-small-parameter-distance} with $F(\zeta) = \frac{2U}{L}\zeta$.
\end{lemma}

\begin{proof}
Note that $f$ satisfying the condition in the lemma satisfies \Cref{ass:f-assumption}. By condition that $f'(1) = f(1) = 0$, we have the following two equalities ($\int_0^x = -\int_x^0 \text{ if } x < 0$):
    \begin{align}\label{eq:integral-alpha-divergence}
\forall x \in \tp{-\frac{1}{2},\frac{1}{2}},\quad        f'(1+x) = \int_0^x f''(1+u) \,du, \quad 
        f(1+x) = \int_0^x(x-u)f''(1+u) \,du.
    \end{align}
Since $0 < L \leq f''(\zeta x) \leq U$ for $\zeta > 1$ and $x \in [-\frac{1}{2\zeta},0) \cup (0,\frac{1}{2\zeta}]$, we have
    \begin{align*}
        |f'(1+\zeta x)| &= \Big| \int_{0}^{\zeta x} f''(1+u) \,du \Big| \leq \zeta |x| \cdot U,\\
        f(1+x) &= \int_0^x (x-u) f''(1+u) \,du \geq L \int_0^x (x-u) \,du = \frac{L x^2}{2}.
    \end{align*}
    Thus, \Cref{cond:general-small-parameter-distance} can be verify as follows
    \begin{align*}
        \frac{xf'(1+\zeta x)}{f(1+x)} = \frac{|x| \cdot |f'(1+\zeta x)|}{f(1+x)} \leq \frac{|x| \cdot \zeta |x| U}{ {L x^2}/{2}} = \frac{2U}{L}\zeta. &\qedhere
    \end{align*}
\end{proof}

Then, we using \Cref{lem:sufficient-condition-for-small-parameter-distance} verify \Cref{cond:general-small-parameter-distance} for Kullback-Leibler, R\'enyi, Jensen-Shannon, $\alpha$-divergence, and Squared Hellinger divergence. The result is summarized as the following table. 
\begin{table}[H]
    \centering
    \begin{tabular}{|c|c|c|c|c|c|}
        \hline
        Divergence & $f(x)$ & $f''(x)$ & $L$ & $U$ & $F(\zeta)$ \\
        \hline
        Kullback-Leibler & $x\ln x - x + 1$ & $\frac{1}{x}$ & $\frac{2}{3}$ & $2$ & $6\zeta$ \\
        R\'enyi  & $-\ln x + x - 1$ & $\frac{1}{x^2}$ & $\frac{4}{9}$ & $4$ & $18\zeta$ \\
        Jensen-Shannon & $\frac{1}{2}(x\ln x - (x+1)\ln\frac{x+1}{2})$ & $ \frac{1}{2x(x+1)}$ & $\frac{2}{15}$ & $\frac{2}{3}$ & $10\zeta$ \\
        $\alpha$-divergence $(\alpha \notin \{0,1\})$ & $\frac{x^\alpha - \alpha x - (1-\alpha)}{\alpha(\alpha-1)}$ & $x^{\alpha-2}$ & $2^{-|\alpha - 2|}$ & $2^{|\alpha-2|}$ & $2\cdot 4^{|\alpha-2|} \zeta$ \\
        Squared Hellinger & $\frac{1}{2}(\sqrt{x}-1)^2$ & $\frac{1}{2\sqrt{x}}$ & $\frac{1}{\sqrt 6}$ & $\frac{1}{\sqrt 2}$ & $2\sqrt{3} \zeta$ \\
        \hline
    \end{tabular}
\end{table}

We remark that \Cref{lem:sufficient-condition-for-small-parameter-distance} does not work for $\chi^\alpha$-divergence but it works for divergences above. 

The algorithm then follows from \Cref{lem:general-small-parameter-distance} such that to approximate $\dvg{\nu}{\mu}{f}$, it only requires polynomial time sampling oracle for the input distribution $\mu$. For squared Hellinger distance, as it is symmetric $\dvg{\nu}{\mu}{f} = \dvg{\mu}{\nu}{f}$, we can swap the roles of $\nu$ and $\mu$ in the algorithm to make sure $\mu$ admits a polynomial time sampling oracle.

\subsection{Algorithms for large parameter distance}\label{sec:alg-large-parameter-distance}
Now, assume that the parameter distance $d_{\mathrm{par}}(\nu,\mu) > \theta = \frac{1}{10(n+3m)}$. By \Cref{lem:f-lower-bound}, the $f$-divergence is at least 
$$\max \left\{f\left(1-\frac{b^2}{2}d_{\mathrm{par}}(\nu,\mu)\right), f\left(1+\frac{b^2}{2}d_{\mathrm{par}}(\nu,\mu)\right) \right\} \geq \max \left\{f\left(1-\frac{b^2}{2}\theta\right), f\left(1+\frac{b^2}{2}\theta\right) \right\},$$ 
where the inequality holds because of the derivative assumptions in \Cref{ass:f-assumption}.
Let $\delta = \frac{b^2}{2}\theta = \frac{1}{\mathrm{poly}(n)}$. We show that all divergences in the above table are at least $\frac{1}{\mathrm{poly}(n)}$. Kullback-Leibler and R\'enyi divergences can be verified by using Taylor series $\ln(1+\delta) = \delta - \frac{\delta^2}{2} + O(\delta^3)$. For Jensen-Shannon, one can rewrite $f(1+\delta) = (1+\delta)\ln(1+\frac{\delta}{2+\delta}) - \ln(1+\frac{\delta}{2})$ and then Taylor series shows $f(1+\delta) = \Omega(\delta^2)$. For $\alpha$-divergence, expanding $(1+\delta)^\alpha$ for real $\alpha$ implies $f(1+\delta) = \frac{\delta^2}{2} \pm O(\delta^3) = \Omega(\delta^2)$. Finally, it is easy to verify that $f(1+\delta) = \Omega(\delta^2)$ for Squared Hellinger divergence. We have 
\begin{align}\label{eq:lower-bound-for-large-parameter-distance}
\dvg{\nu}{\mu}{f} \geq \frac{1}{\mathrm{poly}(n)}.
\end{align}

\paragraph{Kullback-Leibler, R\'enyi, and Jensen-Shannon divergences.} We give the algorithm for Jensen-Shannon divergence. Similar algorithms can be given for KL and R\'enyi divergences. The Jensen-Shannon divergence can be written as 
\begin{align*}
\frac{1}{2}\sum_{\sigma \in \Omega} \nu(\sigma)\ln\frac{\nu(\sigma)}{\mu(\sigma)} - \frac{1}{2}\sum_{\sigma \in \Omega} \nu(\sigma)\ln\frac{\nu(x) + \mu(x)}{2\mu(x)} - \frac{1}{2}\sum_{\sigma \in \Omega} \mu(x) \ln\frac{\nu(x) + \mu(x)}{2\mu(x)}.
\end{align*}
By~\eqref{eq:lower-bound-for-large-parameter-distance}, it suffices to give an algorithm with additive error $\frac{\varepsilon}{\mathrm{poly}(n)}$. We can estimate each term with an additive error. We give the algorithm for the first term. The other two terms can be estimated similarly. We use the sampling oracle to draw $T$ samples $\sigma$ from $\nu$, then using counting oracles to approximate $h(\sigma) =\ln \frac{\nu(\sigma)}{\mu(\sigma)} = \ln(\frac{w_\nu(\sigma)}{w_\mu(\sigma)} \cdot \frac{Z_\mu}{Z_\nu})$, finally, return the average of $h(\sigma)$. The counting oracles estimate the partition functions with a relative error, and thus we can approximate $h(\sigma)$ with an additive error. Due to the marginal lower bound assumption, $|h(\sigma)| \leq n \log \frac{1}{b} = O(n)$. We can set $T = \mathrm{poly}(n,\frac{1}{\varepsilon})$ to achieve the additive error $\frac{\varepsilon}{\mathrm{poly}(n)}$ and our goal is achieved.

\paragraph{$\alpha$-divergence.} The $\alpha$-divergence can be written as
\begin{align*}
    \frac{1}{\alpha(\alpha-1)}\tp{\sum_{\sigma}\frac{\nu(\sigma)^\alpha}{\mu(\sigma)^{\alpha-1}}} - \frac{1}{\alpha(\alpha-1)} = \frac{1}{\alpha(\alpha-1)} \tp{\frac{Z_\mu^{\alpha-1} \cdot Z_\alpha}{Z_\nu^\alpha} - 1},
\end{align*}
where $Z_\alpha$ is the partition function of the Ising model $(G,J^{(\alpha)}, h^{(\alpha)})$. By calling approximate counting oracles with relative error $\e^{ \pm \Theta_\alpha(\delta)}$, we estimate $\frac{Z_\mu^{\alpha-1}Z_\alpha}{Z_\nu^\alpha}$ with $\e^{\pm \delta}$ relative error, which is equivalent to $\pm O(\delta) \frac{Z_\mu^{\alpha-1}Z_\alpha}{Z_\nu^{\alpha}}$ additive error approximation. When the constant $\alpha >1$ or $\alpha < 0$, by \eqref{eq:lower-bound-for-large-parameter-distance}, $\frac{Z_\mu^{\alpha-1}Z_\alpha}{Z_\nu^{\alpha}}$ is lower bounded by $\frac{1}{\mathrm{poly}(n)} + 1$. We set $\delta = \frac{\varepsilon}{\mathrm{poly}(n)}$, then $\frac{Z_\mu^{\alpha-1}Z_\alpha}{Z_\nu^{\alpha}} \geq \frac{\varepsilon}{\varepsilon - \delta}$, which implies
\begin{align*}
\delta \cdot \frac{Z_\mu^{\alpha-1}Z_\alpha}{Z_\nu^{\alpha}} \leq \varepsilon \cdot \tp{\frac{Z_\mu^{\alpha-1}Z_\alpha}{Z_\nu^{\alpha}} - 1}.
\end{align*}
When the constant $0 < \alpha < 1$, $\alpha(1-\alpha) < 0$. By~\eqref{eq:lower-bound-for-large-parameter-distance}, the $\alpha$-divergence at least $\frac{1}{\mathrm{poly}(n)}$, $\frac{Z_\mu^{\alpha-1}Z_\alpha}{Z_\nu^{\alpha}}$ is upper bounded by $1-\frac{1}{\mathrm{poly}(n)}$. We still set $\delta = \frac{\varepsilon}{\mathrm{poly}(n)}$, and in this case,
\begin{align*}
    \delta \cdot \frac{Z_\mu^{\alpha-1}Z_\alpha}{Z_\nu^{\alpha}} \leq \frac{\varepsilon}{\mathrm{poly}(n)} \leq \varepsilon \cdot \tp{1- \frac{Z_\mu^{\alpha-1}Z_\alpha}{Z_\nu^{\alpha}}}.
\end{align*}
In both cases, our algorithm solves the problem with an additive error $O(\varepsilon)\dvg{\nu}{\mu}{\alpha}$ approximation, which implies a relative error $\e^{\pm \varepsilon}$ approximation.

\paragraph{Squared Hellinger distance.} The Squared Hellinger distance can be written as 
\begin{align*}
1 - \sum_{\sigma \in \Omega} \mu(\sigma) \left( \frac{\nu(\sigma)}{\mu(\sigma)} \right)^{1/2}= 1 - \frac{\bar{Z}}{\sqrt{Z_\nu Z_\mu}},
\end{align*}
where $\bar{Z}$ is the partition function of the averaged Ising model in \Cref{thm:squared-hellinger-distance}. 
Again, by~\eqref{eq:lower-bound-for-large-parameter-distance}, it suffices to given an algorithm with additive error $\frac{\varepsilon}{\mathrm{poly}(n)}$.
We can use approximate counting oracles to estimate $\frac{\bar{Z}}{Z_\nu Z_\mu}$ with a relative error. Note that for any $\delta > 0$, $\e^{\pm \delta}$ relative error approximation is the same as $\pm O(\delta)\frac{\bar{Z}}{Z_\nu Z_\mu}$ additive error approximation. Since $\frac{\bar{Z}}{Z_\nu Z_\mu} \leq 1$ (as the divergence is non-negative), we can achieve $\pm O(\delta)$ additive error. Then the whole term $1 - \frac{\bar{Z}}{Z_\nu Z_\mu}$ can be approximated within an additive error $\pm O(\delta)$. The problem is solved by setting $\delta = \frac{\varepsilon}{\mathrm{poly}(n)}$.

\section{Proof of the hardness result} \label{sec:proof-hardness}


In this section, we only consider Ising model $(G,J,h)$ with zero external field and interaction matrix with unified values. Formally, $h = \boldsymbol{0}$ and $J_{uv} = J_{vu} = \frac{\ln \beta}{2}$ for all $\{u,v\} \in E$. It is easy to see the probability of a configuration $\sigma$ is proportional to $\beta^{m(\sigma)}$, where $m(\sigma)$ is the number of monochromatic edges $m(\sigma) = |\{u,v\} \in E, \sigma_u = \sigma_v|$.
We denote it by $(G,\beta)$.

Let $\alpha \geq 2$, $\Delta \geq 3$, and $\beta_\mu > \beta_\nu \geq \frac{\Delta-2}{\Delta}$ such that $(\frac{\beta_\nu}{\beta_\mu})^{\alpha} \beta_\mu < \frac{\Delta-2}{\Delta}$ be constant parameters in \Cref{thm:hardness-without-F}. We define a constant parameter 
\begin{align}\label{eq:beta-for-hardness}
    \beta = \tp{\frac{\beta_\nu}{\beta_\mu}}^\alpha \beta_\mu < \frac{\Delta - 2}{\Delta}.
\end{align}

The following hardness result for approximating partition function $Z=\sum_{\sigma \in \{-1,+1\}^V} \beta^{m(\sigma)}$ of the anti-ferromagnetic Ising model beyond the uniqueness threshold is well-known.
\begin{lemma}[\text{\cite{SS14, GSV16}}]\label{lem:hardness-below-threshold}
    Fix $\Delta \geq 3$ and $0 < \beta < \frac{\Delta - 2}{\Delta}$.
    There exists $c = c(\Delta,\beta) > 0$ such that unless $\mathsf{NP} = \mathsf{RP}$, there is no polynomial time randomized algorithm to approximate the partition function within a factor of $\e^{cn}$ with probability at least $2/3$ for the zero external field Ising model $(G,\beta)$ on $\Delta$-regular $n$-vertex graphs $G$.
\end{lemma}


For any $\Delta$-regular graph $G$, let $Z(G,\beta)$ denote the partition function of the zero external field Ising model $(G,\beta)$, where $\beta$ is defined in \eqref{eq:beta-for-hardness}. We prove the following lemma.
\begin{lemma}\label{lem:hardness-for-chi-alpha-divergence}
Let $\Delta,\alpha,\beta_\nu,\beta_\mu$ be constants. There exists a constant $C = C(\Delta,\alpha,\beta_\nu,\beta_\mu) > 0$ such that for any graph $G$, let $\mu$ and $\nu$ be the Gibbs distributions of the Ising models $(G,\beta_\nu)$ and $(G,\beta_\mu)$,
\begin{align*}
 \frac{1}{C} \cdot Z(G,\beta)  \leq  \dvg{\nu}{\mu}{\chi^\alpha} \cdot \frac{Z(G,\beta_\nu)^\alpha}{Z(G,\beta_\mu)^{\alpha-1}} \leq  C \cdot Z(G,\beta).
\end{align*}
\end{lemma}

Assume \Cref{lem:hardness-for-chi-alpha-divergence} holds. We prove \Cref{thm:hardness-without-F}.

\begin{proof}[Proof of \Cref{thm:hardness-without-F}]
Fix parameters $\Delta,\alpha,\beta_\nu,\beta_\mu$. Let $\beta$ be defined in \eqref{eq:beta-for-hardness}.

Suppose there exists an FPRAS for $\dvg{\nu}{\mu}{\chi^\alpha}$. By run algorithms independently and take median, we can approximate $\dvg{\nu}{\mu}{\chi^\alpha}$ within the factor of $2$ in polynomial time with probability at least $0.99$. By our assumption, $\beta_\nu,\beta_\mu \geq \frac{\Delta - 2}{\Delta}$. We can compute $Z(G,\beta_\nu)$ within the factor of $2$ by applying the algorithms in~\cite{CCYZ25} (if $\beta_\nu \in [\frac{\Delta - 2}{\Delta}, 1]$) or the algorithms in~\cite{JS93} (if $\beta_\nu > 1$) in polynomial time, where the algorithm succeeds with probability at least $0.99$. Similarly, we can compute an approximation to $Z(G,\beta_\mu)$.

By \Cref{lem:hardness-for-chi-alpha-divergence}, we get a constant approximation to $Z(G,\beta)$ in polynomial time with probability at least $0.99$.  The hardness result follows from \Cref{lem:hardness-below-threshold}.
\end{proof}

The rest of this section is devoted to prove \Cref{lem:hardness-for-chi-alpha-divergence}. To prove the lemma, we consider a middle term  $\sum_{\sigma \in \Omega} \mu(\sigma) ( \frac{\nu(\sigma)}{\mu(\sigma)} + 1 )^\alpha$. The following result is a corollary of \Cref{lem:error-bound}.
\begin{corollary}\label{cor:error-bound-for-chi-alpha-divergence}
    There exists $C' = C'(\Delta,\alpha,\beta_\nu,\beta_\mu)$ such that
    \begin{align*}
            \frac{1}{C'} \cdot \sum_{\sigma \in \{\pm\}^V} \mu(\sigma) \left( \frac{\nu(\sigma)}{\mu(\sigma)} + 1 \right)^\alpha \leq \dvg{\nu}{\mu}{\chi^\alpha} \leq C' \cdot \sum_{\sigma \in \{\pm\}^V} \mu(\sigma) \left( \frac{\nu(\sigma)}{\mu(\sigma)} + 1 \right)^\alpha.
    \end{align*}
    \end{corollary}
    \begin{proof}
    The second inequality is trivial. For the first inequality, by \Cref{lem:error-bound}, we can set $\theta$ in the lemma to be $\frac{1}{2} \left( \ln \beta_\nu - \ln \beta_\mu \right)$, which is a constant depending on $\beta_\nu,\beta_\mu$.
    Furthermore, the underlying graph $G$ has constant degree $\Delta$, both $\beta_\mu$ and $\beta_\nu$ are constants. Hence, by the conditional independence property of the Ising model, both $\nu$ and $\mu$ have a constant marginal lower bound $b = b(\Delta,\beta_\nu,\beta_\mu)$. 
    Hence $B_{\alpha,b}(\theta) = \Theta_{\alpha,\beta_\nu,\beta_\mu,\Delta}(1)$ is a constant.
    \end{proof}
    With \Cref{cor:error-bound-for-chi-alpha-divergence}, \Cref{lem:hardness-for-chi-alpha-divergence} is a straightforward corollary of the following lemma.
    \begin{lemma}
        It holds that
        \begin{align*}
            {Z(G,\beta)} \leq \frac{Z(G,\beta_\nu)^\alpha}{Z(G,\beta_\mu)^{\alpha-1}} \sum_{\sigma \in \{\pm\}^V} \mu(\sigma) \left( \frac{\nu(\sigma)}{\mu(\sigma)} + 1 \right)^\alpha \leq 3^\alpha \cdot Z(G,\beta).
        \end{align*}
    \end{lemma}
 \begin{proof}
 To simplify the notation, we fix the graph $G$ in the proof. We denote the partition function $Z(G,\beta_\nu)=Z_\nu$ and $Z(G,\beta_\mu)=Z_\mu$. The family of Ising models in $\+F(\nu,\mu,\alpha)$ in~\eqref{eq:family-of-isings} can be written as $(G,\beta_k)$, where $\beta_k = \frac{\beta_\nu^k}{\beta_\mu^{k-1}}$ for $0 \leq k \leq \alpha$. We use $Z_k$ to denote the partition function of $(G,\beta_k)$. We remark that $Z_0 = Z_\mu$, $Z_1 = Z_\nu$, $\beta$ in~\eqref{eq:beta-for-hardness} is the same as $\beta_\alpha$, and $Z(G,\beta) = Z_\alpha$.

 We use the binomial expansion that
 \begin{align}\label{eq:expansion-for-domination}
     \sum_{\sigma \in \{\pm\}^V} \mu(\sigma) \left( \frac{\nu(\sigma)}{\mu(\sigma)} + 1 \right)^\alpha = \sum_{k=0}^\alpha \binom{\alpha}{k} \sum_{\sigma \in \{\pm\}^V} \frac{\nu^k(\sigma)}{\mu^{k-1}(\sigma)} = \sum_{k=0}^\alpha \binom{\alpha}{k} \frac{Z_\mu^{k-1}}{Z_\nu^k} \cdot Z_k.
 \end{align}
 We claim the following inequality holds. For any $0 \leq k \leq \alpha - 1$,
 \begin{align}\label{eq:domination-inequality}
     \frac{Z_\mu^k}{Z_\nu^{k+1}}Z_{k+1} \geq \frac{1}{2} \frac{Z_\mu^{k-1}}{Z_\nu^k}Z_k.
 \end{align}
 We prove \eqref{eq:domination-inequality} later. With this conclusion, combining with \eqref{eq:expansion-for-domination}, we know that 
 \begin{align}\label{eq:bounded-domination-factor}
     \frac{Z_\mu^{\alpha-1}}{Z_\nu^\alpha} Z_\alpha \leq \sum_{k=0}^\alpha \binom{\alpha}{k} \frac{Z_\mu^{k-1}}{Z_\nu^k} \cdot Z_k \leq \frac{Z_\mu^{\alpha-1}}{Z_\nu^\alpha} Z_\alpha \cdot \sum_{k=0}^\alpha \binom{\alpha}{k} 2^{\alpha-k} = 3^\alpha \frac{Z_\mu^{\alpha-1}}{Z_\nu^\alpha} Z_\alpha,
 \end{align}
 where the first inequality is trivial and the second inequality is due to \eqref{eq:domination-inequality}. Rearranging the terms in~\eqref{eq:bounded-domination-factor}, we get the desired result.

 Finally, we prove \eqref{eq:domination-inequality}. Recall that $m(\sigma)$ denotes the number of monochromatic edges under the configuration $\sigma$.  Since $\beta_\mu > \beta_\nu$, the sequence $\beta_k$ is monotonically decreasing as $\beta_k = \frac{\beta_\nu^k}{\beta_\mu^{k-1}}$. 
 Fix an unordered pair of configurations $\{\sigma,\tau\} \in \{\pm\}^V \times \{\pm\}^V$ (it may be possible that $\sigma = \tau$). Without loss of generality, we assume $m(\sigma) \geq m(\tau)$, otherwise we can swap the two configurations.

We have the following inequality for all $1 \leq k \leq \alpha - 1$, 
 \begin{align}\label{eq:domination-inequality-for-beta-k}
    &\beta_\mu^{m(\sigma)} \beta_{k+1}^{m(\tau)} + \beta_\mu^{m(\tau)} \beta_{k+1}^{m(\sigma)} \overset{(\textsf{I})}{\geq} \beta_\mu^{m(\sigma)} \beta_{k+1}^{m(\tau)}  \overset{(\textsf{II})}{=} \left( \frac{\beta_\mu}{\beta_\nu} \right)^{m(\sigma)} \beta_\nu^{m(\sigma)} \cdot \left( \frac{\beta_\nu}{\beta_\mu} \right)^{m(\tau)} \beta_k^{m(\tau)}\notag\\ 
    = &\left( \frac{\beta_\mu}{\beta_\nu} \right)^{m(\sigma) - m(\tau)} \beta_\nu^{m(\sigma)} \beta_k^{m(\tau)}
    \overset{(\textsf{III})}{\geq} \beta_\nu^{m(\sigma)} \beta_k^{m(\tau)} \overset{(\textsf{IV})}{\geq} \frac{1}{2} \left( \beta_\nu^{m(\sigma)} \beta_k^{m(\tau)} + \beta_\nu^{m(\tau)} \beta_k^{m(\sigma)} \right),
\end{align}
where the inequality $(\textsf{I})$ throws the second term, and the equality $(\textsf{II})$ is due to the definition of $\beta_k$, the inequality $(\textsf{III})$ is due to the assumption $m(\sigma) \geq m(\tau)$ and $\beta_\mu > \beta_\nu$, and the last inequality $(\textsf{IV})$ is due $m(\sigma) \geq m(\tau)$ and $\beta_{\nu} = \beta_1 \geq \beta_k$ for $k \geq 1$.

If $k = 0$, note that $\beta_0 = \beta_\mu$ and $\beta_1 = \beta_\nu$, we have
\begin{align}\label{eq:domination-inequality-for-beta-0}
    \beta_\mu^{m(\sigma)} \beta_{k+1}^{m(\tau)} + \beta_\mu^{m(\tau)} \beta_{k+1}^{m(\sigma)} = \beta_\nu^{m(\sigma)} \beta_k^{m(\tau)} + \beta_\nu^{m(\tau)} \beta_k^{m(\sigma)} \geq \frac{1}{2} \left( \beta_\nu^{m(\sigma)} \beta_k^{m(\tau)} + \beta_\nu^{m(\tau)} \beta_k^{m(\sigma)} \right).
\end{align}

To prove~\eqref{eq:domination-inequality},  we need to prove $Z_\mu Z_{k+1} \geq \frac{1}{2} Z_\nu Z_k$ for all $0 \leq k \leq \alpha - 1$. Recall that the partition function $Z_k = \sum_{\sigma \in \{\pm\}^V} \beta_k^{m(\sigma)}$. We have the following inequality
\begin{align*}
    Z_\mu Z_{k+1} &=  \sum\limits_{\sigma,\tau \in \{\pm\}^V} \beta_\mu^{m(\sigma)} \beta_{k+1}^{m(\tau)} 
     = \sum_{\substack{\text{unordered pair}\{\sigma,\tau\}\\ \sigma \neq \tau}} \tp{\beta_\mu^{m(\sigma)}\beta_{k+1}^{m(\tau)} + \beta_\mu^{m(\tau)}\beta_{k+1}^{m(\sigma)}} + \sum_{\sigma \in \{\pm\}^V} \beta_\mu^{m(\sigma)}\beta_{k+1}^{m(\sigma)}\\
 (\ast) \quad   &\geq \frac{1}{2}  \sum_{\substack{\text{unordered pair}\{\sigma,\tau\}\\ \sigma \neq \tau}} \tp{\beta_\nu^{m(\sigma)}\beta_{k}^{m(\tau)} + \beta_\nu^{m(\tau)}\beta_{k}^{m(\sigma)}} + \frac{1}{2}\sum_{\sigma \in \{\pm\}^V} \beta_\nu^{m(\sigma)}\beta_{k}^{m(\sigma)}  = \frac{1}{2} Z_\nu Z_k,
\end{align*}
where the inequality $(\ast)$ is due to \eqref{eq:domination-inequality-for-beta-k} and \eqref{eq:domination-inequality-for-beta-0} and two inequalities \eqref{eq:domination-inequality-for-beta-k} and \eqref{eq:domination-inequality-for-beta-0} hold even if $\sigma = \tau$.
This proves the inequality in~\eqref{eq:domination-inequality}.
\end{proof}

\ifthenelse{\boolean{DoubleBlind}}{}{
    \section*{Acknowledgements}
    Weiming Feng is supported by the Early Career Scheme of the Hong Kong Research Grants Council under grant number 27202725.
    Yucheng Fu was supported by the 2025 Summer Research Internship Programme in the School of Computing and Data Science at The University of Hong Kong. 
}
    


    
    


\bibliographystyle{alpha}
\bibliography{refs}
\end{document}